\documentclass[sigconf, nonacm, pdfa]{acmart}
\usepackage[a-2b]{pdfx}

\usepackage{amssymb}
\usepackage{enumitem}
\usepackage{amsthm}
\newtheorem{theorem}{Theorem}[section]

\newtheorem{lemma}[theorem]{Lemma}
\theoremstyle{definition}
\newtheorem{example}[theorem]{Example}
\newtheorem*{assumption}{Assumption}
\usepackage{balance}
\usepackage{adjustbox}
\usepackage{etoolbox}
\newcommand{\showme}{yes}

\newcommand{\fullversion}[1]{%
  \ifdefstring{\showme}{yes}{#1}{}%
}
\newcommand{\shortversion}[1]{%
  \ifdefstring{\showme}{no}{#1}{}%
}

\newcommand{\Changed}[1]{#1}

\usepackage[binary-units]{siunitx}
\DeclareSIUnit{\txn}{txn}
\DeclareSIUnit{\batch}{batch}
\graphicspath{{./DrawIO/}}

\newcommand{\n}{\mathbf{n}}
\newcommand{\f}{\mathbf{f}}

\newcommand{\Name}[1]{\textnormal{\textsc{#1}}}
\newcommand{\PBFT}{\Name{Pbft}}

\newcommand{\RCC}{\Name{RCC}}

\newcommand{\Tusk}{\Name{Tusk}}

\newcommand{\MName}[1]{\textsc{#1}}

\newcommand{\digest}[1]{d_{#1}}

\newcommand{\twat}{\textit{txns\_w\_assigned\_oi}}

\newcommand{\hts}{highest\_oi\_list}

\newcommand{\lotm}{local\_timer}

\usepackage[noend]{algorithmic}
\usepackage{algorithm,float}


\usepackage{algorithmic}

\newcommand{\GETS}{:=}
\newenvironment{myprotocol}{
    \hrule
    \small
    \smallskip
    \algsetup{linenosize=\footnotesize}
    \begin{algorithmic}[1]
        
        \newcommand{\SPACE}{\item[]}
        \newcommand{\TITLE}[2]{\item[] \textbf{\underline{##1}} (##2) \textbf{:}\\[2pt]}
        \makeatletter
            \newcommand{\EVENT}[1]{\STATE \textbf{event} ##1 \textbf{do}\begin{ALC@g}}
            \newcommand{\ENDEVENT}{\end{ALC@g}}
        \makeatother
        
        \makeatletter
            \newcommand{\FUNCTION}[2]{\STATE \textbf{function} \Name{##1}(##2) \textbf{do}\begin{ALC@g}}
            \newcommand{\ENDFUNCTION}{\end{ALC@g}}
        \makeatother
}{
    \end{algorithmic}%
    \hrule
}

\usepackage{tikz,pgfplots,pgfplotstable}

\usetikzlibrary{shapes,arrows,automata,positioning,cd}
\tikzset{
  spotlessedge/.style   = {black, ->, >=stealth},
}
\usepackage{xcolor}

\usetikzlibrary{arrows.meta,decorations.pathreplacing}
\tikzset{
    >=Stealth,
    smalltext/.append style={scale=0.7},
    dot/.style={circle,scale=0.35,draw=black,fill=black}
}

\usetikzlibrary{automata, positioning, arrows}
\usetikzlibrary{arrows.meta}
\tikzset{
    >=Stealth,
    plot/.append style={baseline,scale=0.475},
    label/.append style={font=\strut\footnotesize},
    dot/.style={circle,scale=0.35,draw=black,fill=black},
}
\pgfplotscreateplotcyclelist{mycyclelist}{
    black   ,every mark/.append style={solid,fill=\pgfplotsmarklistfill},mark=*\\ 
    red     ,every mark/.append style={solid,fill=\pgfplotsmarklistfill},mark=square*\\
    blue    ,every mark/.append style={solid,fill=\pgfplotsmarklistfill},mark=triangle*\\
    teal    ,every mark/.append style={solid,fill=\pgfplotsmarklistfill},mark=pentagon*\\
    brown   ,every mark/.append style={solid,fill=\pgfplotsmarklistfill},mark=halfsquare*\\ 
    orange  ,every mark/.append style={solid,fill=\pgfplotsmarklistfill},mark=halfcircle*\\
    violet  ,every mark/.append style={solid,fill=\pgfplotsmarklistfill,rotate=180},mark=halfdiamond*\\
}
\pgfplotscreateplotcyclelist{mycyclelistex}{
    black   ,every mark/.append style={solid,fill=\pgfplotsmarklistfill},mark=*\\ 
    red     ,every mark/.append style={solid,fill=\pgfplotsmarklistfill},mark=square*\\
}
\pgfplotsset{
    tick label style={font=\large},
    legend style={font=\Large,cells={anchor=west}},
    title style={font=\Large},
    label style={font=\Large},
    width=262.5pt,
    height=185pt,
    every axis/.append style={
        ylabel near ticks,
        xlabel near ticks,
        mark size=2.5pt,
        cycle list name=mycyclelist,
        font=\Large,
        y tick label style={
                        /pgf/number format/precision=1,
                        /pgf/number format/fixed,
                        /pgf/number format/fixed zerofill
                    }
    },
    barstyle/.append style={
        ybar,
        bar width={0.5cm},
        enlarge x limits=0.4,
        enlarge y limits={upper=0.025},
        ymin=0,
        xtick=data
    }
}

\newcommand{\FDAG}{\Name{FairDAG}}
\newcommand{\FairDAG}{\Name{FairDAG-AB}}
\newcommand{\FairDAGRL}{\Name{FairDAG-RL}}

\newcommand{\Pompe}{\Name{Pompe}}
\newcommand{\Themis}{\Name{Themis}}

\newcommand{\Rashnu}{\Name{Rashnu}}
\usepackage{tikz,pgfplots,pgfplotstable}
\usetikzlibrary{arrows.meta}
\tikzset{
    >=Stealth,
    node_text/.append style={font=\strut\bfseries},
    plot/.append style={baseline,scale=0.8},
}
\pgfplotscreateplotcyclelist{mycyclelist}{
    thick,black   ,every mark/.append style={solid,fill=\pgfplotsmarklistfill},mark=*\\
    thick,red     ,every mark/.append style={solid,fill=\pgfplotsmarklistfill},mark=square*\\
    thick,blue    ,every mark/.append style={solid,fill=\pgfplotsmarklistfill},mark=triangle*\\
    thick,teal    ,every mark/.append style={solid,fill=\pgfplotsmarklistfill},mark=pentagon*\\
    thick,brown   ,every mark/.append style={solid,fill=\pgfplotsmarklistfill},mark=halfsquare*\\
    thick,orange  ,every mark/.append style={solid,fill=\pgfplotsmarklistfill},mark=halfcircle*\\
    thick,violet  ,every mark/.append style={solid,fill=\pgfplotsmarklistfill,rotate=180},mark=halfdiamond*\\
}
\pgfplotscreateplotcyclelist{mycyclelistex}{
    black   ,every mark/.append style={solid,fill=\pgfplotsmarklistfill},mark=*\\
    red     ,every mark/.append style={solid,fill=\pgfplotsmarklistfill},mark=square*\\
}

\pgfplotsset{
 tick label style={font=\Large},
 legend style={font=\LARGE,cells={anchor=west}},
 title style={font=\Large},
 label style={font=\LARGE},
 width=262.5pt,
 height=185pt,
every axis/.append style={
        xlabel near ticks,
        mark size=2.5pt,
        cycle list name=mycyclelist,
        y tick label style={
            }
    },
barstyle/.append style={
        ybar,
       bar width={0.5cm},
        enlarge x limits=0.4,
        enlarge y limits={upper=0.025},
        ymin=0,
        xtick=data
    }
every axis/.append style={
ylabel near ticks,
xlabel near ticks,
mark size=1pt,
cycle list name=mycyclelist,
 font=\Large,
enlargelimits=0.1
},
}

\pgfplotsset{select coords index/.style n args={2}{
    x filter/.code={
        \ifnum\coordindex<#1\fi
        \ifnum\coordindex>#2\fi
    }
}}

\definecolor{colAA}{RGB}{255,87,51}
\definecolor{colAB}{RGB}{255,116,51}
\definecolor{colBA}{RGB}{53,142,250}
\definecolor{colBB}{RGB}{147,199,251}
\definecolor{colCA}{RGB}{92,208,48}
\definecolor{colCB}{RGB}{167,252,105}
\definecolor{colDA}{RGB}{214,217,60}
\definecolor{colDB}{RGB}{242,245,52}
\definecolor{colEA}{RGB}{242,84,228}
\definecolor{colEB}{RGB}{223,128,216}
\definecolor{colFA}{RGB}{159,28,249}
\definecolor{colGA}{RGB}{28,71,249}
\definecolor{colHA}{RGB}{153,187,221}

\pgfplotstableread{
Id  Nodes  Throughput Latency 
2	16	1932885.584	4.969254573
3	24	1897078.833	6.847341625
4	32	1802312.67	9.961781897
5	64	1754272.241	13.12646176
}\dataTuskRegion

\pgfplotstableread{
Id  Nodes  Throughput Latency 
2	16	555732.8972	18.54887393
3	24	505194.0559	24.29347549
4	32	403173.9506	25.93701288
5	64	257427.02	29.84453479
}\dataPBFTRegion

\pgfplotstableread{
Id  Nodes  Throughput Latency 
2	16	1742134.118	6.221912123
3	24	1660754.848	7.912748506
4	32	1623646.775	11.96989385
5	64	1390355.247	15.32789827
}\dataFairABRegion

\pgfplotstableread{
Id  Nodes  Throughput Latency 
2	16	338047.1759	60.61112183
3	24	229864.9408	64.78280174
4	32	98320.5073	77.72147327
5	64	77110.78898	85.24080635
}\dataPomepeRegion

\pgfplotstableread{
Id  Nodes  Throughput Latency 
2	16	1995747.731	6.628858507
3	24	1964151.734	8.812348406
4	32	1930200.36	13.67652324
5	64	1905581.264	23.8288249
}\dataRCCRegion

\pgfplotstableread{
Id  Nodes  Throughput Latency 
2	16	10306.8531	71.91652671
3	24	9579.76037	80.40461857
4	32	8421.62712	94.93983875
5	64	7800.93381	106.62450725
}\dataThemisRegion

\pgfplotstableread{
Id  Nodes  Throughput Latency 
2	16	10306.8531	71.91652671
3	24	9579.76037	80.40461857
4	32	8421.62712	94.93983875
5	64	7800.93381	106.6245073
}\dataThemisRegion

\pgfplotstableread{			
Id  Nodes  Throughput Latency 			
2	16	35053.33333	65.96357625
3	24	31861.70213	77.38590333
4	32	12510.48309	87.40741857
5	64	8176.954315	94.87910192
}\dataFairRLRegion

\pgfplotstableread{
Id  Nodes  Throughput Latency 
1	16	1156530.938	0.7353069
2	24	1149409.167	0.9625515227
3	32	1120227.083	0.9828967708
4	40	1102916.091	1.0517498409
}\dataPerfTusk

\pgfplotstableread{
Id  Nodes  Throughput Latency 
1	16	924983.125	1.131248042
2	24	851221.0606	1.245942659
3	32	778669.4271	1.370004375
4	40	654183.9167	1.521305354
}\dataPerfPBFT

\pgfplotstableread{
Id  Nodes  Throughput Latency 
1	16	1179415.417	0.8527709083
2	24	1160903.333	0.9320103409
3	32	1141493.385	0.9613873542
4	40	1101220.773	1.2405721818
}\dataPerfRCC

\pgfplotstableread{
Id  Nodes  Throughput Latency 
1	16	854408.9773	1.416703636
2	24	802423.9394	1.442484773
3	32	690239.7917	1.643138333
4	40	546206.5	2.163223636
}\dataPerfPompe

\pgfplotstableread{
Id  Nodes  Throughput Latency			
1	16	993109.1667	0.2767572778
2	24	918059.0152	0.3154469091
3	32	833129.7917	0.650009
4	40	722614	0.8076409583
}\dataPerfFairAB

\pgfplotstableread{
Id  Node Throughput Latency 
1	16	170045.2083	2.2335486382978713
2	24	133348.9394	2.888793409090909
3	32	97704.27083	3.4945071666666665
4	40	80077.47323	5.744034318181818
}\dataPerfThemis

\pgfplotstableread{
Id  Nodes  Throughput Latency 
1	16	176445.1042	1.200770479
2	24	159868.1818	1.592416364
3	32	153539.6354	2.354203688
4	40	145030.1212	4.071525455
}\dataPerfFairRL

\pgfplotstableread{
Id  Regions  Throughput Latency 
1	1	1874837.865	10.8041359
2	2	1851501.892	11.18165139
3	4	1825786.056	11.72174258
4	6	1802312.67	12.961781897
}\dataTuskRegions

\pgfplotstableread{
Id  Regions  Throughput Latency 
1	1	213025.625	57.18824202
2	2	137723.709	69.32110142
3	4	107192	73.695
4	6	98320.5073	77.72147327
}\dataPompeRegions

\pgfplotstableread{
Id  Regions  Throughput Latency 
1	1	1630446.602	7.959914398
2	2	1629883.573	8.398158206
3	4	1628274.056	9.095127489
4	6	1623646.775	11.96989385
}\dataFairRegions

\pgfplotstableread{
Id  Regions  Throughput Latency 
1	1	1463341.761	15.46101973
2	2	502086.9516	24.99898406
3	4	368780.1064	27.58144596
4	6	257427.02	29.84453479
}\dataPBFTRegions

\pgfplotstableread{
Id  Regions  Throughput Latency 
1	1	2056280.106	8.866952677
2	2	1986098.535	13.78896763
3	4	1959613.643	14.52642759
4	6	1930200.36	15.67652324
}\dataRCCRegions

\pgfplotstableread{			
Id  Regions  Throughput Latency 		
1	1	332367	64.4179243
2	2	28203	73.44959153
3	4	17280	80.10652371
4	6	16453	87.40741857
}\dataFairRLRegions

\pgfplotstableread{
Id  Regions  Throughput Latency 
1	1	9265.70313	81.33861843
2	2	9053.83918	82.07788103
3	4	8584.08353	83.10652371
4	6	8421.62712	94.93983875
}\dataThemisRegions

\pgfplotstableread{			
Id  Nodes  Throughput Latency 			
1	16	126259.3939	5.431405545
2	24	103768.5507	6.6689535
3	32	64261.17302	8.603072727
}\dataPompeFail

\pgfplotstableread{			
Id  Nodes  Throughput Latency 			
1	16	1253400	0.3078039182
2	24	1019585.29	0.3950533417
3	32	476500	0.5581836842
}\dataFairABFail

\pgfplotstableread{			
Id  Nodes  Throughput Latency 			
1	16	7903.804878	16.30429809
2	24	6932.965517	19.758085474
3	32	6664.43038	21.31978607
}\dataThemisFail

\pgfplotstableread{			
Id  Nodes  Throughput Latency 			
1	16	117740.3636	7.239816602
2	24	85734.615385	17.12096357
3	32	5553.522267	21.18604995
}\dataFairRLFail

\pgfplotstableread{	
Id  Nodes  ratio  			
1  consensus	59.6816150
2  computation 40.31838493
}\dataPompeBD	

\pgfplotstableread{	
Id  Nodes  ratio  			
1 consensus	91.18615821	
2 computation	08.813841787
}\dataFairABBD	

\pgfplotstableread{	
Id  Nodes  ratio  			
1 consensus	99.91263862	
2 computation	00.08736138187
}\dataThemisBD	

\pgfplotstableread{	
Id  Nodes  ratio  			
1 consensus	99.02870682	
2 computation	00.9712931803
}\dataFairRLBD

\pgfplotstableread{
Id  Node Throughput Latency 
1	16	1240785.795	0.1133875636
2	24	1188439.722	0.1932088083
3	32	912466.1458	0.234295521
}\dataFairABConsensus

\pgfplotstableread{
Id  Node Throughput Latency 
1	16	1932537.813	0.42063295
2	24	1731447.424	0.6364171
3	32	609791.6667	1.858013367
4	64	284331.2386	3.036717092
}\dataThemisConsensus

\pgfplotstableread{
Id  Node Throughput Latency 
1	16	1268455.114	0.1153927909
2	24	1109621.597	0.149007625
3	32	890737.3438	0.1973538667
}\dataFairRLConsensus

\pgfplotstableread{
Id  Node Throughput Latency 
1	16	1762118.864	5.896453807
2	24	1813027.727	8.94082197
3	32	1667222.38	11.62885848
4	64	1481424.485	20.59183731
}\dataPompeDelay

\pgfplotstableread{
Id  Node Throughput Latency 
1	16	532524.0217	17.17450567
2	24	23116.0396	20.30450769
3	32	16514.61654	18.57251103
4	64	7680.862745	17.13019146
}\dataThemisDelay

\newcommand{\AllTimeline}{\begin{tikzpicture}[plot]
    \begin{axis}[
            xlabel={Timeline (s)},
            ylabel={Throughput (txn/s)},
            legend to name={tllegend},
            legend style={font=\footnotesize, only marks, legend columns=-1},
            grid=major,
            width=550,
            height=160,
            ytick={0, 200000, 400000, 600000, 800000, 1000000},
            y tick label style={
                /pgf/number format/.cd,
                precision=0,
                scaled y ticks=base 10:-5,
                fixed,
                /tikz/.cd,
                font=\scriptsize
            },
        ]
        \addplot[color=orange,   thick, mark=diamond*] table[x={Id},y={Pompe}] {\dataTimeline};
        \addplot[color=red,   thick, mark=diamond*] table[x={Id},y={FairAB}] {\dataTimeline};
        \addplot[color=colFA,   thick, mark=oplus*] table[x={Id},y={Themis}] {\dataTimeline};
        \addplot[color=colHA,   thick, mark=oplus*] table[x={Id},y={FairRL}] {\dataTimeline};

        \legend{\Pompe{}, \FairDAG{}, \Themis{}, \FairDAGRL{}}
    \end{axis}
\end{tikzpicture}}

\pgfplotstableread{
Id  Pompe Themis FairRL FairAB
1	0	0	0	0
2	0	0	0	0
3	520560	76970	166940	715100
4	676580	282780	272660	830640
5	626300	299170	273540	864600
6	660200	249880	284060	877260
7	781740	61150	276380	877500
8	32140	0	329820	869580
9	0	0	260980	875620
10	0	0	302380	849980
11	0	0	274740	699040
12	0	0	347160	805980
13	0	0	334520	875140
14	0	0	270940	873640
15	0	108740	308020	862960
16	209360	275760	310020	774400
17	653540	288600	281960	865920
18	675580	299880	313780	898840
19	683280	290920	314200	863480
20	689420	286900	286620	861620
21	693240	274700	319040	874020
22	677900	286700	262820	866820
23	677400	274480	326480	864520
24	691920	285980	319400	857840
25	688040	293720	290880	869980
}\dataTimeline

\pgfplotstableread{
Id  F Throughput Latency 
1	0	600239.7917	0.293138333
2	2	580235.3333	0.32459
3	5	570304.5062	0.375838958
4	10	592968.7121 0.38638875
}\dataPompeMFail

\pgfplotstableread{
Id  F Throughput Latency 
1	0	883129.7917	0.320009
2	2	868523.3333	0.34548875
3	5	852635.303	0.38316875
4	10	786820.5303	0.429666139
}\dataFairABMFail

\pgfplotstableread{
Id  F Throughput Latency 
1	0	70004.0404	0.262850875
2	2	66410.10101	0.307860833
3	5	55190.60606	0.325071111
}\dataThemisMFail

\pgfplotstableread{
Id  F Throughput Latency 
1	0	133849.8909	0.2157800455
2	2	127024.2424	0.255576556
3	5	121434.4624	0.279894167
4	10	120338.4524	0.3038825
}\dataFairRLMFail

\newcommand{\FairPerfF}[7]{\begin{tikzpicture}[plot]
    \begin{axis}[
            xlabel={#3},
            ylabel={#1},
            xtick={1,2,3,4},
            xticklabels={5,6,7,8},
            ytick={#5},
            ymin=0,
            scaled y ticks = true,
            legend to name={ftpslegend},
            legend style={font=\footnotesize, only marks, legend columns=-1},
            grid=major,
            width=300,
            height=180,
            title={#7},
            y tick label style={
                /pgf/number format/.cd,
                scaled y ticks=base 10:#4,
                precision=#6,
                fixed,
                /tikz/.cd,
            },
        ]

        \addplot[color=black,   thick, mark=*] table[x={Id},y={#2}] {\dataFPBFT};
        \addplot[color=blue,    thick, mark=triangle*] table[x={Id},y={#2}] {\dataFRCC};
        \addplot[color=green!50!black,   thick, mark=triangle*] table[x={Id},y={#2}] {\dataFTusk};
        \addplot[color=orange,   thick, mark=diamond*] table[x={Id},y={#2}] {\dataFPompe};
        \addplot[color=red,   thick, mark=diamond*] table[x={Id},y={#2}] {\dataFFairAB};
        \addplot[color=colFA,   thick, mark=oplus*] table[x={Id},y={#2}] {\dataFThemis};
        \addplot[color=colHA,   thick, mark=oplus*] table[x={Id},y={#2}] {\dataFFairRL};
        \legend{PBFT, RCC, Tusk, \Pompe{}, \FairDAG{}, \Themis{}, \FairDAGRL{}}
    \end{axis}
\end{tikzpicture}}

\newcommand{\FairPerfFTPS}{\FairPerfF{Throughput (txn/s)}{Throughput}{$\f$}{0}{}{1}{(a) Single-Region}}
\newcommand{\FairPerfFLAT}
{\FairPerfF{Latency (s)}{Latency}{$\f$}{0}{0, 0.1, 0.2, 0.3, 0.4}{1}{\Changed{(b) Single-Region}}}

\pgfplotstableread{
Id F  Throughput Latency 
1 5	732112.7083	0.187446875
2 6	700797.0175	0.20089375
3 7	655902.1488	0.2308563636
4 8	608405.9636	0.2947217955
}\dataFPompe

\pgfplotstableread{
Id F  Throughput Latency 
1	5	987285.9375	0.2017647083
2	6	951231.866	0.2257058182
3	7	926106.5289	0.2651751818
4	8	904629.0909	0.3213259545
}\dataFFairAB

\pgfplotstableread{
Id F  Throughput Latency 
1	5	401138  0.060277
2	6	358716  0.0637286
3	7	324851  0.0679553
4	8	294645	0.07374
}\dataFThemis

\pgfplotstableread{
Id F  Throughput Latency 
1	5	431345	0.0485739
2	6	375950	0.053978
3	7	343348	0.056334
4	8	309574	0.0609983
}\dataFFairRL

\pgfplotstableread{
Id F  Throughput Latency 
1	5	1181688.125	0.1549558333
2	6	1108022.456	0.1692632917
3	7	1082462.893	0.1762419773
4	8	1065460.873	0.2078955909
}\dataFTusk

\pgfplotstableread{
Id F  Throughput Latency 
1	5	897901.5625	0.1350189583
2	6	857601.3158	0.1403918958
3	7	805959.7521	0.1492265455
4	8	749865.7455	0.1567201136
}\dataFPBFT

\pgfplotstableread{
Id F  Throughput Latency 
1	5	1185831.146	0.1407449938
2	6	1139046.316	0.1599698277
3	7	1082701.322	0.1873888636
4	8	1056844.218	0.2199540455
}\dataFRCC

\newcommand{\FairTL}[5]{\begin{tikzpicture}[plot]
    \begin{axis}[
            xlabel={#3},
            ylabel={#1},
            ytick={0.1, 0.2, 0.3, 0.4, 0.5,0.6,0.7,0.8,0.9,1.0,1.1,1.2},
            legend to name={tlbdlegend},
            legend style={font=\footnotesize, only marks, legend columns=4},
            grid=major,
            width=600,
            height=200,
            ymin=0,
            y tick label style={
                /pgf/number format/.cd,
                scaled y ticks=base 10:#4,
                fixed,
                /tikz/.cd,
            },
        ]
        \addplot[color=black,   thick, mark=*] table[x={Throughput},y={Latency}] {\dataTLPBFT};     
    \addplot[color=blue,    thick, mark=triangle*  ] table[x={Throughput},y={Latency}] {\dataTLRCC};      
    \addplot[color=green!50!black,   thick, mark=triangle*] table[x={Throughput},y={Latency}] {\dataTLTUSK};     
    \addplot[color=orange,   thick, mark=diamond* ] table[x={Throughput},y={Latency}] {\dataTLPompe};    
    \addplot[color=red,   thick, mark=diamond*     ] table[x={Throughput},y={Latency}] {\dataTLFairAB};   
    \addplot[color=colFA,   thick, mark=oplus*   ] table[x={Throughput},y={Latency}] {\dataTLThemis};   
    \addplot[color=colHA,   thick, mark=oplus*   ] table[x={Throughput},y={Latency}] {\dataTLFairRL};   
        \legend{\PBFT{}, \RCC{}, Tusk, \Pompe{}, \FairDAG{}, \Themis{}, \FairDAGRL{}}
    \end{axis}
\end{tikzpicture}}

\newcommand{\FairPerfTL}{\FairTL{Latency (s)}{Latency (s)}{Throughput (txn/s)}{0}{0, 50000, 100000, 150000, 200000, 500000,1000000, 1200000}}

\pgfplotstableread{
Throughput  Latency
    325800      0.1423442333
	608405.9636	0.2947217955
    620239.7917 0.4643138333
    632423.9394 0.51823
    635452.6545 0.5192526364
    636455.9273 0.8650431591
}\dataTLPompe

\pgfplotstableread{
Throughput  Latency
    512036.3636 0.1881871009
    904629.0909	0.3213259545
    900300.073 0.4080276591
    903832.8 0.60897675
    911145.7143  0.7500224643
    918204.727 0.8599173864
    904130.255 0.9605231136
}\dataTLFairAB

\pgfplotstableread{
Throughput  Latency
    138472 0.0356015430
    218374 0.0472734
    294645 0.07374
    302834 0.18374
    299374 0.48573
    301239 0.90384
}\dataTLThemis

\pgfplotstableread{
Throughput  Latency
    164927 0.01949
    203848 0.04128
    309574 0.06099
    313837 0.16047
    313374 0.45138
    320384 0.85931
}\dataTLFairRL

\pgfplotstableread{
Throughput  Latency
388850.9091	0.06091611364
556323.3766	0.07404021136
742404.9351	0.1256704545
733552.1212	0.3359768409
743169.8701	0.7487404318
743127.013	1.008497523
}\dataTLPBFT

\pgfplotstableread{
Throughput  Latency
176772.7273	0.1535560909
310819.0476	0.1803677273
1065460.873	0.2078955909
1058857.576	0.2562105455
1067060.606	0.5886297955
1071507.446	0.9709388864
}\dataTLTUSK

\pgfplotstableread{
Throughput  Latency
114876.5368	0.0567338636
867792.121	0.1169998409
1056844.218	0.2199540455
1082354.459	0.2480175682
1081490.303	0.5691926818
1082676.537	0.9320204318
}\dataTLRCC

\pgfplotstableread{
x1	x2	FD0	FD1	FD7	FD10	TH0	TH1	TH7	TH10
1 1 0.955 0.831 0.802 0.800 0.779 0.767 0.754 0.751
3 3 0.984 0.962 0.846 0.853 0.897 0.921 0.835 0.772
5 5 0.993 0.985 0.908 0.916 0.970 0.963 0.861 0.834
7 7 1.000 0.987 0.959 0.922 0.990 0.979 0.917 0.889
9 9 1.000 0.987 0.973 0.957 0.998 1.000 0.946 0.884
11 11 1.000 0.992 0.988 0.962 1.000 0.998 0.973 0.937
13 13 0.996 0.992 0.990 0.983 1.000 1.000 0.965 0.929
15 15 1.000 0.995 0.991 0.982 1.000 1.000 0.984 0.952
17 17 1.000 0.995 0.997 0.979 1.000 0.997 0.984 0.973
19 19 1.000 0.995 0.995 0.990 1.000 1.000 0.998 0.987
21 21 1.000 0.996 1.000 0.986 1.000 1.000 1.000 0.998
23 23 1.000 0.997 0.998 0.990 1.000 1.000 1.000 0.992
25 25 1.000 0.998 0.999 0.997 1.000 1.000 1.000 0.998
27 27 1.000 0.999 1.000 1.000 1.000 1.000 0.997 0.997
29 29 1.000 1.000 0.999 0.997 1.000 1.000 1.000 1.000
31 31 1.000 1.000 1.000 1.000 1.000 1.000 1.000 1.000
31 33 1.000 1.000 1.000 1.000 1.000 1.000 1.000 1.000
31 35 1.000 1.000 1.000 1.000 1.000 1.000 1.000 1.000
31 37 1.000 1.000 1.000 1.000 1.000 1.000 1.000 1.000
31 39 1.000 1.000 1.000 1.000 1.000 1.000 1.000 1.000
31 41 1.000 1.000 1.000 1.000 1.000 1.000 1.000 1.000
}\dataReverseOriginal

\pgfplotstableread{
x1	x2	FD0	FD1	FD7	FD10	TH0	TH1	TH7	TH10
3.226 2.390 0.955 0.831 0.802 0.800 0.779 0.767 0.754 0.751
9.677 7.317 0.984 0.962 0.846 0.853 0.897 0.921 0.835 0.772
16.129 12.195 0.993 0.985 0.908 0.916 0.970 0.963 0.861 0.834
22.613 17.073 1.000 0.987 0.959 0.922 0.990 0.979 0.917 0.889
29.032 21.951 1.000 0.987 0.973 0.957 0.998 1.000 0.946 0.884
35.484 26.829 1.000 0.992 0.988 0.962 1.000 0.998 0.973 0.937
41.935 31.709 0.996 0.992 0.990 0.983 1.000 1.000 0.965 0.929
48.387 36.585 1.000 0.995 0.991 0.982 1.000 1.000 0.984 0.952
54.839 41.463 1.000 0.995 0.997 0.979 1.000 0.997 0.984 0.973
61.290 46.341 1.000 0.995 0.995 0.990 1.000 1.000 0.998 0.987
67.742 51.220 1.000 0.996 1.000 0.986 1.000 1.000 1.000 0.998
74.194 56.098 1.000 0.997 0.998 0.990 1.000 1.000 1.000 0.992
80.645 60.976 1.000 0.998 0.999 0.997 1.000 1.000 1.000 0.998
87.097 65.854 1.000 0.999 0.998 1.000 1.000 1.000 1.000 0.997
93.548 70.732 1.000 1.000 0.999 0.997 1.000 1.000 1.000 1.000
100.000 75.610 1.000 1.000 1.000 1.000 1.000 1.000 1.000 1.000
100.000 80.488 1.000 1.000 1.000 1.000 1.000 1.000 1.000 1.000
100.000 85.366 1.000 1.000 1.000 1.000 1.000 1.000 1.000 1.000
100.000 90.244 1.000 1.000 1.000 1.000 1.000 1.000 1.000 1.000
100.000 95.122 1.000 1.000 1.000 1.000 1.000 1.000 1.000 1.000
100.000 100.000 1.000 1.000 1.000 1.000 1.000 1.000 1.000 1.000
}\dataReverse

\pgfplotstableread{
x2	x1	FD10 TH10 FD7 TH7 FD1 TH1 FD0 TH0
3.225806452	2.43902439	0.421053	0.0925926	0.363636	0.333333	0.761905	0.707692	0.894737	0.666667
9.677419355	7.317073171	0.44	0.173913	0.512195	0.5	0.964286	0.888889	1	0.846154
16.12903226	12.19512195	0.5	0.224138	0.71875	0.441176	0.978261	0.910448	1	0.958333
22.58064516	17.07317073	0.55814	0.344444	0.723404	0.552632	1	0.896552	1	0.987179
29.03225806	21.95121951	0.475	0.284211	0.948718	0.777778	0.982759	1	1	1
35.48387097	26.82926829	0.608696	0.398148	0.928571	0.846939	1	0.987654	1	1
41.93548387	31.70731707	0.84	0.487179	0.974359	0.858586	1	1	1	1
48.38709677	36.58536585	0.864407	0.519685	1	0.928571	1	1	1	1
54.83870968	41.46341463	0.857143	0.745283	1	1	1	1	1	1
61.29032258	46.34146341	0.982456	0.810345	1	1	1	1	1	1
67.74193548	51.2195122	1	0.913043	1	0.989583	1	1	1	1
74.19354839	56.09756098	1	0.940171	1	1	1	1	1	1
80.64516129	60.97560976	1	0.953271	1	1	1	1	1	1
87.09677419	65.85365854	1	1	1	1	1	1	1	1
93.5483871	70.73170732	1	1	1	1	1	1	1	1
100	75.6097561	1	1	1	1	1	1	1	1
100	80.48780488	1	1	1	1	1	1	1	1
100	85.36585366	1	1	1	1	1	1	1	1
100	90.24390244	1	1	1	1	1	1	1	1
100	95.12195122	1	1	1	1	1	1	1	1
100	100	1	1	1	1	1	1	1	1
}\datafairnessquality

\pgfplotstableread{
x2	x1	FD10 TH10 FD7 TH7 FD1 TH1 FD0 TH0
1	1	0.421053	0.0925926	0.363636	0.333333	0.761905	0.707692	0.894737	0.666667
3	3	0.44	0.173913	0.512195	0.5	0.964286	0.888889	1	0.846154
5	5	0.5	0.224138	0.71875	0.441176	0.978261	0.910448	1	0.958333
7	7	0.55814	0.344444	0.723404	0.552632	1	0.896552	1	0.987179
9	9	0.475	0.284211	0.948718	0.777778	0.982759	1	1	1
11	11	0.608696	0.398148	0.928571	0.846939	1	0.987654	1	1
13	13	0.84	0.487179	0.974359	0.858586	1	1	1	1
15	15	0.864407	0.519685	1	0.928571	1	1	1	1
17	17	0.857143	0.745283	1	1	1	1	1	1
19	19	0.982456	0.810345	1	1	1	1	1	1
21	21	1	0.913043	1	0.989583	1	1	1	1
23	23	1	0.940171	1	1	1	1	1	1
25	25	1	0.953271	1	1	1	1	1	1
27	27	1	1	1	1	1	1	1	1
29	29	1	1	1	1	1	1	1	1
31	31	1	1	1	1	1	1	1	1
31	33	1	1	1	1	1	1	1	1
31	35	1	1	1	1	1	1	1	1
31	37	1	1	1	1	1	1	1	1
31	39	1	1	1	1	1	1	1	1
31	41	1	1	1	1	1	1	1	1
}\datafairnessqualityoriginal

\pgfplotstableread{
x	Pompe10   FD10      FD7 Pompe7   FD1 Pompe1   FD0	Pompe0	
1	0.532	0.522	0.529	0.555	0.605	0.608	0.609	0.602
2	0.569	0.595	0.635	0.640	0.774	0.765	0.755	0.767
3	0.637	0.692	0.736	0.726	0.861	0.855	0.839	0.853
4	0.678	0.824	0.828	0.789	0.930	0.907	0.912	0.916
5	0.676	0.906	0.897	0.856	0.967	0.955	0.963	0.956
6	0.744	0.934	0.937	0.900	0.984	0.985	0.980	0.979
7	0.786	0.966	0.973	0.942	0.992	0.994	0.992	0.991
8	0.885	0.983	0.986	0.966	0.998	0.999	0.997	0.997
9	0.905	0.992	0.997	0.988	0.999	1.000	0.999	0.999
10	0.910	0.999	1.000	0.998	1.000	1.000	1.000	1.000
11	0.947	0.999	0.999	0.996	1.000	1.000	1.000	1.000
12	0.972	1.000	1.000	0.999	1.000	1.000	1.000	1.000
13	0.987	1.000	1.000	1.000	1.000	1.000	1.000	1.000
14	0.997	1.000	1.000	1.000	1.000	1.000	1.000	1.000
15	0.994	1.000	1.000	1.000	1.000	1.000	1.000	1.000
16	1.000	1.000	1.000	1.000	1.000	1.000	1.000	1.000
17	1.000	1.000	1.000	1.000	1.000	1.000	1.000	1.000
18	1.000	1.000	1.000	1.000	1.000	1.000	1.000	1.000
19	1.000	1.000	1.000	1.000	1.000	1.000	1.000	1.000
20	1.000	1.000	1.000	1.000	1.000	1.000	1.000	1.000
}\datafairnessqualityabsolute

\pgfplotstableread{
x	Pompe10   FD10       Pompe7 FD7   Pompe1 FD1   	Pompe0	FD0
1	0.542	0.543	0.542	0.544	0.621	0.575	0.625	0.64
2	0.672	0.666	0.653	0.648	0.767	0.82	0.777	0.84
3	0.727	0.748	0.757	0.744	0.861	0.91	0.872	0.907
4	0.779	0.801	0.83	0.836	0.923	0.954	0.926	0.956
5	0.818	0.847	0.887	0.903	0.963	0.97	0.961	0.97
6	0.856	0.871	0.928	0.947	0.982	0.984	0.985	0.981
7	0.889	0.898	0.96	0.974	0.993	0.99	0.994	0.995
8	0.919	0.916	0.982	0.99	0.998	0.994	0.998	0.989
9	0.948	0.956	0.992	0.997	1	0.997	0.999	0.999
10	0.973	0.983	0.998	0.999	1	0.999	1	1
11	0.984	0.99	0.998	1	1	0.999	1	1
12	0.987	0.993	0.999	1	1	0.999	1	0.999
13	0.99	0.993	1	1	1	0.999	1	1
14	0.991	0.994	0.999	0.999	1	1	1	1
15	0.993	0.996	1	0.999	1	1	1	1
16	0.994	0.996	1	0.999	1	1	1	1
17	0.996	0.997	1	0.999	1	1	1	1
18	1	0.997	1	1	1	1	1	1
19	1	0.998	1	1	1	1	1	1
20	1	0.998	1	1	1	1	1	1
}\datareverseabsolute

\newcommand{\fairnessqualityfigure}[6]{\begin{tikzpicture}[plot, scale=1]
    \begin{axis}[xlabel={#3},ylabel={#4},title={#2},xtick={#5}, ymin={#6}, height=180,  legend to name={qualitylegend},legend columns=4,
        grid=major]
        \addplot[very thick, color=blue] table[x={x1},y={TH0}] {#1};
        \addplot[very thick, color=orange] table[x={x1},y={TH1}] {#1};
        \addplot[very thick, color=black] table[x={x1},y={TH7}] {#1};
        \addplot[very thick, color=green!50!black] table[x={x1},y={TH10}] {#1};
        \addplot[very thick, dashed, color=blue] table[x={x2},y={FD0}] {#1};
        \addplot[very thick, dashed, color=orange] table[x={x2},y={FD1}] {#1};
        \addplot[very thick, dashed, color=black] table[x={x2},y={FD7}] {#1};
        \addplot[very thick, dashed, color=green!50!black] table[x={x2},y={FD10}] {#1};
        \legend{Themis-0, Themis-1, Themis-7, Themis-10, FairDAG-RL-0, FairDAG-RL-1, FairDAG-RL-7, FairDAG-RL-10}
    \end{axis}
\end{tikzpicture}}

\newcommand{\fairnessqualityfigureabsolute}[6]{\begin{tikzpicture}[plot, scale=1]
    \begin{axis}[xlabel={#3},ylabel={#4}, height=180, title={#2},xtick={#5}, ymin={#6}, legend to name={qualitylegendab},legend columns=4,
        grid=major]
        \addplot[very thick, color=blue] table[x={x},y={Pompe0}] {#1};
        \addplot[very thick, color=orange] table[x={x},y={Pompe1}] {#1};
        \addplot[very thick, color=black] table[x={x},y={Pompe7}] {#1};
        \addplot[very thick, color=green!50!black] table[x={x},y={Pompe10}] {#1};
        \addplot[very thick, dashed, color=blue] table[x={x},y={FD0}] {#1};
        \addplot[very thick, dashed, color=orange] table[x={x},y={FD1}] {#1};
        \addplot[very thick, dashed, color=black] table[x={x},y={FD7}] {#1};
        \addplot[very thick, dashed, color=green!50!black] table[x={x},y={FD10}] {#1};
        \legend{Pompe-0, Pompe-1, Pompe-7, Pompe-10, FairDAG-AB-0, FairDAG-AB-1, FairDAG-AB-7, FairDAG-AB-10}
    \end{axis}
\end{tikzpicture}}

\newcommand{\rashnufigure}[6]{\begin{tikzpicture}[plot, scale=1]
    \begin{axis}[
    xlabel={#3},
    ylabel={#4},
    title={#2},
    xtick={#5}, 
    legend to name={rashnulegend},
    legend columns=4,
    ymin={0},
    xticklabels={0.01,0.5,0.99},
    height=160,
    y tick label style={
        /pgf/number format/.cd,
        scaled y ticks=true,
        precision=0,
        fixed,
        /tikz/.cd,
    },
        grid=major]
        \addplot[color=colFA,   thick, mark=oplus*] table[x={s},y={themis}] {#1};
        \addplot[color=colHA,   thick, mark=oplus*] table[x={s},y={fairdag}] {#1};
        \addplot[color=colFA,   thick, mark=diamond*] table[x={s},y={rthemis}] {#1};
        \addplot[color=colHA,   thick, mark=diamond*] table[x={s},y={rfairdag}] {#1};
        \legend{Themis, FairDAG-RL, Themis-Rashnu, FairDAG-RL-Rashnu}
    \end{axis}
\end{tikzpicture}}

\pgfplotstableread{
s themis fairdag rthemis rfairdag		
0.01 294645 309574 387362 472282
0.5 294645 309574  376649 435299
0.99 294645 309574 315423 316370
}\datarashnutput

\pgfplotstableread{
s themis fairdag rthemis rfairdag	
0.01 73.74	60.99 52.268 44.384
0.5	 73.74	60.99 54.945 48.862
0.99 73.74	60.99 65.332 54.425
}\datarashnulat

\newcommand{\axisticksquality}{0, 10, 20, 30, 40, 50, 60, 70, 80, 90, 100}
\newcommand{\axisticksweight}{1, 6, 11, 16, 21, 26, 31, 36, 41}

\newcommand{\axistickdiff}{1, 6, 11, 16, 20}

\pgfplotstableread{
ID SK Throughput  TputA
1 0.01    6491  6533
2 0.5    6491   6533
3 1    6491     6533
4 2    6491     6533
}\dataskThemis

\pgfplotstableread{
ID SK Throughput  TputA
1 0.01    9775 41945
2 0.5    9775 41945
3 1    9775 41945
4 2    9775 41945
}\dataskFairRL

\pgfplotstableread{
ID SK Throughput  TputA
1 0.01	131174	119508
2 0.5	122626	119533
3 1	86787   118990
4 2	31729   108930
}\dataskRUDAG

\pgfplotstableread{
ID SK Throughput TputA
1 0.01	12476 9909
2 0.5	11622 9520
3 1	10039 8028
4 2	7020 6831
}\dataskRUThemis

\newcommand{\FairRegion}[7]{\begin{tikzpicture}[plot]
    \begin{axis}[
            xlabel={\Changed{#3}},
            ylabel={#1},
            xtick={1,2,3,4},
            xticklabels={1,2,3,4},
            ytick={#5},
            ymin={0},
            title={\Changed{#7}},
            legend to name={regionlegend},
            legend style={font=\footnotesize, only marks, legend columns=4},
            grid=major,
            width=300,
            height=180,
            scaled y ticks = true,
            y tick label style={
                /pgf/number format/.cd,
                scaled y ticks=base 10:#4,
                precision=#6,
                fixed,
                /tikz/.cd,
            },
        ]

        \addplot[color=black,   thick, mark=*] table[x={Id},y={#2}] {\datRegionPbft};
        \addplot[color=blue,    thick, mark=triangle*] table[x={Id},y={#2}] {\datRegionRCC};
        \addplot[color=green!50!black,   thick, mark=triangle*] table[x={Id},y={#2}] {\datRegionTusk};
        \addplot[color=orange,   thick, mark=diamond*] table[x={Id},y={#2}] {\datRegionPompe};
        \addplot[color=red,   thick, mark=diamond*] table[x={Id},y={#2}] {\datRegionFairAB};
        \addplot[color=colFA,   thick, mark=oplus*] table[x={Id},y={#2}] {\datRegionThemis};
        \addplot[color=colHA,   thick, mark=oplus*] table[x={Id},y={#2}] {\datRegionFairRL};
    \end{axis}
\end{tikzpicture}}

\pgfplotstableread{	
Id	region	Throughput	Latency
1	1	946906.6667	1.288258017
2	2	885161.3818	5.256281351
3	3	837743.6364	6.258671981
4	4	794419.7818	7.10568864
}\datRegionFairAB

\pgfplotstableread{	
Id	region	Throughput	Latency
1	1	791003.2667	0.6967938
2	2	729477.6	3.509099099
3	3	698336.0727	4.020342523
4	4	684389.12	5.525442178
}\datRegionPompe

\pgfplotstableread{	
Id	region	Throughput	Latency	
1	1	1206005.333	0.6549763833
2	2	1170752.667	1.537107
3	3	1178753.6	1.707260583
4	4	1168347.733	2.62885417
}\datRegionTusk

\pgfplotstableread{	
Id	region	Throughput	Latency	
1	1	1197402.2	0.6841030025
2	2	1183047.933	2.018710075
3	3	1172612.667	2.698525975
4	4	1164310.6	3.076023917
}\datRegionRCC

\pgfplotstableread{	
Id	region	Throughput	Latency		
1	1	849697.1004	0.3590039464
2	2	791979.8658	0.9502405083
3	3	751812.8	1.030092558
4	4	730946.5763	1.22955775
}\datRegionPbft

\pgfplotstableread{	
Id	region	Throughput	Latency		
1	1	332367  0.092024
2	2	85663   0.99934
3	3	53343   1.3049
4	4	50283   1.4738
}\datRegionFairRL

\pgfplotstableread{	
Id	region	Throughput	Latency		
1	1	283893  0.1394
2	2	50490  1.872
3	3	37550  2.6659
4	4	32746  2.9374
}\datRegionThemis

\newcommand{\FairRegionTPS}{\FairRegion{Throughput (txn/s)}{Throughput}{Number of Regions}{0}{}{1}{(c) Multi-Region, $\f=8$}}
\newcommand{\FairRegionLAT}{\FairRegion{Latency (s)}{Latency}{Number of Regions}{0}{}{0}{(d) Multi-Region, $\f=8$}}

\newcommand\vldbdoi{10.14778/3773749.3773763}
\newcommand\vldbpages{265-278}
\newcommand\vldbvolume{19}
\newcommand\vldbissue{2}
\newcommand\vldbyear{2025}

\newcommand\vldbtitle{\shorttitle} 
\newcommand\vldbavailabilityurl{https://github.com/apache/incubator-resilientdb/tree/fairdag}
\newcommand\vldbpagestyle{empty}

\begin{document}
\title{FairDAG: Consensus Fairness over Multi-Proposer Causal Design}

\author{Dakai Kang, Junchao Chen, Tien Tuan Anh Dinh$^{\dagger}$, Mohammad Sadoghi}
\affiliation{
\institution{Exploratory Systems Lab, University of California, Davis}
	\institution{$\dagger$Deakin University}
}

\begin{abstract}
The rise of cryptocurrencies like Bitcoin and Ethereum has driven interest in blockchain database technology, with smart contracts enabling the growth of decentralized finance (DeFi). However, research has shown that adversaries exploit transaction ordering to extract profits through attacks like front-running, sandwich attacks, and liquidation manipulation. This issue affects blockchains where block proposers have full control over transaction ordering. To address this, a more fair transaction ordering mechanism is essential.

Existing fairness protocols, such as \Pompe{} and \Themis{}, operate on leader-based consensus protocols, which not only suffer from low throughput caused by single-leader bottleneck, but also give adversarial block proposers to manipulate transaction ordering. To address these limitations, we propose a new framework \FDAG{} that runs fairness protocols on top of DAG-based consensus protocols, which improves protocol performance in both throughput and fairness quality, leveraging the multi-proposer design and validity property of DAG-based consensus protocols.

We conducted a comprehensive analytical and experimental evaluation of two \FDAG{} variants—\FairDAG{} and \FairDAGRL{}. Our results demonstrate that \FDAG{} outperforms prior fairness protocols in both throughput and fairness quality.

\end{abstract}

\maketitle

\pagestyle{\vldbpagestyle}
\begingroup\small\noindent\raggedright\textbf{PVLDB Reference Format:}\\
Dakai Kang, Junchao Chen, Tien Tuan Anh Dinh, Mohammad Sadoghi. \vldbtitle. PVLDB, \vldbvolume(\vldbissue): \vldbpages, \vldbyear.\\
\href{https://doi.org/\vldbdoi}{doi:\vldbdoi}
\endgroup
\begingroup
\renewcommand\thefootnote{}\footnote{\noindent
This work is licensed under the Creative Commons BY-NC-ND 4.0 International License. Visit \url{https://creativecommons.org/licenses/by-nc-nd/4.0/} to view a copy of this license. For any use beyond those covered by this license, obtain permission by emailing \href{mailto:info@vldb.org}{info@vldb.org}. Copyright is held by the owner/author(s). Publication rights licensed to the VLDB Endowment. \\
\raggedright Proceedings of the VLDB Endowment, Vol. \vldbvolume, No. \vldbissue\ %
ISSN 2150-8097. \\
\href{https://doi.org/\vldbdoi}{doi:\vldbdoi} \\
}\addtocounter{footnote}{-1}\endgroup

\ifdefempty{\vldbavailabilityurl}{}{
\vspace{.3cm}
\begingroup\small\noindent\raggedright\textbf{PVLDB Artifact Availability:}\\
The source code, data, and/or other artifacts have been made available at \url{\vldbavailabilityurl}.
\endgroup
}

\section{Introduction}\label{sec:intro}

The emergence of cryptocurrencies, including Bitcoin~\cite{bitcoin} and Ethereum~\cite{ethereum}, has sparked broad interest in blockchain database technology~\cite{mysten, solana, aptos, chainlink}. Blockchain enables a new class of applications, namely decentralized finance (DeFi)~\cite{defi,defismartcontract,defieth,defiartical, heimbach2022risks}, whose market capitalization exceeds $70$ billion. DeFi requires consistency and fairness in transaction ordering. The former ensures that all participants agree on the same transaction order, which has been addressed under crash-failure settings of traditional distributed databases, for example~\cite{caerus, lowlatency, dynamast, calvin, slog}, and under Byzantine-failure settings in blockchains and verifiable databases~\cite{ethereum, blockchainmeetsdb,spitz} where Byzantine participants can have arbitrary malicious behavior. In the presence of Byzantine participants, even though transaction ordering is consistent, its fairness remains vulnerable to ordering manipulation attacks. Such an “order manipulation crisis” is possible because block proposers have full control over transaction selection and ordering. Byzantine proposers can censor or reorder transactions to extract \emph{Maximal Extractable Value (MEV)} from blocks~\cite{flashboysfrontrunning, yang2022sok, intime, richricher,fairnessmatters, zhang2024no, malkhi2022maximal, weintraub2022flash}, which is unfair to other participants. Attacks include front-running, back-running, sandwich attacks, liquidation manipulation, and time-bandit attacks~\cite{flashboysfrontrunning,sandwichattacks,liquiditymanipulation,orderattacks,qin2022quantifying}.

The key to achieving fair transaction ordering lies in preventing proposers from dominating the ordering. Existing studies~\cite{wendyfairness,themis,pompe,rashnu, quickorder,constantinescu2023fair,mu2024separation} have proposed various fairness protocols. Unlike traditional protocols where each block contains a list of transactions, in fairness protocols, each block contains local orderings from a set of participants.
Once blocks are committed, a final transaction ordering is derived from the local orderings. Different fairness protocols guarantee different fairness properties, reflecting the preferences of correct participants who honestly report the order in which they receive transactions. For example, \Pompe{}~\cite{pompe} calculates the \emph{assigned ordering indicator} for each transaction and orders transactions based on it. It guarantees that transaction $T_1$ is ordered before $T_2$ if every correct participant receives $T_1$ before any correct participant receives $T_2$, a property named \emph{Ordering Linearizability}. \Themis{}~\cite{themis} constructs a dependency graph among transactions and determines the ordering according to the edges of the graph. It guarantees that if a $\gamma$ proportion of correct participants receive $T_1$ before $T_2$, then $T_1$ will be ordered \emph{no later} than $T_2$ -a property named \emph{$\gamma$-Batch-Order-Fairness}. 

We observe that a well-designed fairness protocol should achieve the following goals: 

\begin{enumerate}[label=G\arabic*]
    \item \textbf{Resilience to Ordering Manipulation.} The protocol should limit Byzantine participants’ influence on the final transaction order to preserve fairness properties.
    \item \textbf{Minimal Correct Participants Requirement.} The protocol should preserve fairness properties while relying on as few correct participants as possible.
    \item \textbf{High Performance.} Fairness protocols should minimize the overhead introduced by fair ordering, achieving high throughput and low latency.
\end{enumerate}

Unfortunately, previous fairness protocols~\cite{aequitas, themis,pompe, quickorder, rashnu, wendyfairness} leverage leader-based consensus protocols~\cite{pbftj,hotstuff,poe}, relying on a single leader to collect local orderings from other participants. A Byzantine leader can manipulate transaction ordering by selectively collecting local orderings to maximize its \emph{MEV}. Moreover, even without Byzantine behavior, a single leader may become a performance bottleneck when the workload exceeds its capacity, as the message complexity of broadcasting collected local orderings grows quadratically with the number of participants.

To address the challenges posed by the underlying leader-based consensus protocols, we propose \FDAG{}, a novel framework that runs fairness protocols on top of DAG-based consensus protocols~\cite{narwhal,dagrider,bullshark, shoal, shoalpp, mysticeti}, which provide the following features that can enhance ordering fairness:

\begin{itemize}
    \item \textbf{Multi-Proposer High-throughput Design}: DAG-based protocols allow all participants to propose blocks in parallel. This approach improves system performance by alleviating the bottleneck introduced by a single leader.
    \item \textbf{Validity through Causal Design}: Blocks in DAG-based protocols reference blocks from other participants, forming a \emph{Directed Acyclic Graph (DAG)}. The causal relationship in DAG guarantees that vertices from correct participants are eventually committed by all correct participants.
\end{itemize}

The \emph{Multi-Proposer High-throughput Design} removes the bottleneck caused by a single leader (\textbf{G3}). The \emph{validity} constrains the Byzantine participants' ability to selectively collect local orderings (\textbf{G1}), and thus lowers the requirement on the number of correct participants (\textbf{G2}).

In \FDAG{}, each participant proposes its block containing the local ordering as a DAG vertex and reliably broadcasts it, and a final transaction ordering is deterministically generated based on the vertices committed by the underlying DAG-based consensus protocols. We designed two variants of \FDAG{}, namely \emph{absolute-ordering} \FairDAG{} and \emph{relative-ordering} \FairDAGRL{}.

There are two main challenges in employing DAG-based protocols as the underlying consensus layer for fairness. First, to mitigate the high latency of DAG-based consensus protocols, participants should leverage uncommitted DAG vertices to reduce the latency for fairness decisions, but participants may hold inconsistent views of uncommitted vertices. Second, Byzantine participants may attempt to manipulate the ordering by selectively ignoring vertices proposed by specific participants.
\FDAG{} addresses these challenges through a novel ordering indicator manager, adaptive fairness thresholds, and new DAG construction rules that collectively ensure fairness.

We make the following contributions: 
\begin{enumerate}
    \item We propose \FairDAG{}, a \emph{absolute} fairness protocol that guarantees \emph{Ordering Linearizability}. We propose a \emph{Ordering Indicator Manager} and an adaptive fairness threshold called \emph{LPAOI} that are compatible with the multi-proposer design and commit rules of the DAG-based consensus protocols.
    \item We propose \FairDAGRL{}, a \emph{relative} fairness protocol that guarantees \emph{$\gamma$-batch-order-fairness}. Leveraging the validity property of DAG-based consensus protocols, we adopt new thresholds for dependency graph construction to improve system performance and reduce the requirement of minimal correct participant number.
    \item In \FDAG{}, we apply new rules for constructing a DAG to guarantee fairness against adversarial participants.
    \item We conducted comprehensive analytical and experimental evaluation of our protocols. The results show that: compared to \Pompe{} and \Themis{}, \FairDAG{} and \FairDAGRL{} outperform in both throughput and fairness quality, which reflects the fairness protocols' resilience against adversarial ordering manipulations.
\end{enumerate}

This paper is structured as follows. Section~\ref{sec:prelim} presents the background. Section~\ref{sec:model} introduces the system model. Section~\ref{sec:overviewfairdag} \Changed{presents an overview of \FDAG{} protocols. Section~\ref{sec:design} and Section~\ref{sec:design2} describe the details of \FairDAG{} and \FairDAGRL{}\shortversion{\footnote{\Changed{see Section 8 in our extended report~\cite{extended-report} for correctness proofs}}.}} Section~\ref{sec:newcomparison} analytically compares \FDAG{} with prior fairness protocols \Pompe{} and \Themis{}. \fullversion{Section~\ref{sec:proof} presents the correctness proof.} Section~\ref{sec:eval} presents the experimental evaluation of \FDAG{} and baseline protocols. Section~\ref{sec:related} discusses other related works, and Section~\ref{sec:conclusion} concludes. 
\section{Background}\label{sec:prelim}

\FairDAG{} and \FairDAGRL{} execute fairness protocols atop DAG-based consensus protocols. Beginning with this section, our discussion is framed within the context of \emph{Byzantine Fault Tolerant (BFT)} protocols, where participants are referred to as replicas. \emph{Byzantine replicas}, corrupted by an adversary, may exhibit arbitrary malicious behavior, whereas the remaining \emph{correct replicas} behave benignly.
In this section, we present the definitions of three fairness properties and introduce the DAG-based consensus protocols.

\begin{figure}[t]
    \begin{tikzpicture}[yscale=0.6]
        \node at (-2,1) {$R_1: \{T_1, T_2, T_3, T_4\}$};
        \node at (-2,0.4) {$R_2: \{T_2, T_3, T_4, T_1\}$};
        \node at (-2,-0.2) {$R_3: \{T_3, T_4, T_1, T_2\}$};
        \node at (-2,-0.8) {$R_4: \{T_4, T_1, T_2, T_3\}$};
    
        \node[minimum width=0.4cm,draw,circle,font=\tiny] (a) at (0.4,0) {$T_1$};
        \node[minimum width=0.4cm,draw,circle,font=\tiny] (c) at (2,0.8) {$T_2$};
        \node[minimum width=0.4cm,draw,circle,font=\tiny] (d) at (3.6,0) {$T_3$};
        \node[minimum width=0.4cm,draw,circle,font=\tiny] (f) at (2,-0.8) {$T_4$};
        
        \draw[->] (a) -- (c);
        \draw[->] (c) -- (d);
        \draw[->] (d) -- (f);
        \draw[->] (f) -- (a);
    \end{tikzpicture}
    \caption{Condorcet Cycle}
    \label{fig:cycle}
\end{figure}

\subsection{Receive-Order-Fairness}\label{ssec:receiveorder}

\begin{definition}\label{def:Receive-Order-Fairness}
\textbf{Receive-Order-Fairness.} For any two transactions $T_1$ and $T_2$, if all correct replicas receive $T_1$ before $T_2$, then $T_1$ must be ordered \emph{before} $T_2$ in the final ordering.
\end{definition}

We show that it is impossible to always guarantee \emph{Receive-Order-Fairness} in the presence of Byzantine replicas for \emph{Condorcet Cycles}. 
Figure~\ref{fig:cycle} illustrates this impossibility with a concrete example. Consider four replicas ${R_1, R_2, R_3, R_4}$, among which at most one may be Byzantine. Let there be four transactions ${T_1, T_2, T_3, T_4}$. As shown on the left side of the figure, for any two commands $T_i$ and $T_{i+1}$ (modulo 4) , three replicas receive $T_i$ before $T_{i+1}$ but the fourth differs. Since the identity of the Byzantine replica is unknown, we respect all majority-endorsed orderings supported by at least three replicas when determining the final ordering. The right side of Figure~\ref{fig:cycle} depicts a directed graph where an edge $T_i \rightarrow T_j$ indicates that at least three replicas received $T_i$ before $T_j$. The resulting graph contains a Condorcet Cycle~\cite{brandt2016handbook,condorcet1785essay}—a cycle of pairwise preferences that cannot be linearly extended without violating at least one of them. Hence, the \emph{Receive-Order-Fairness} is impossible in this case.

\subsection{Ordering Linearizability}\label{ssec:Linearizability}

\emph{Ordering Linearizability} is a fairness property \Changed{introduced} by \Pompe{}~\cite{pompe} and \Changed{adopted} by our protocol, \FairDAG{}. To achieve this property, each replica assigns a monotonically increasing \emph{ordering indicator} \Changed{(denoted {$oi$}) to each transaction reflecting the order it receives the transactions.} The final ordering is then derived from ordering indicators collected from a majority of replicas.

Denoting by $ois_i^C$ the set of ordering indicators for the transaction $T_i$ from the correct replicas, we define \emph{Ordering Linearizability}:

\begin{definition}\label{def:Linearizability}
\textbf{Ordering Linearizability.} For any two transactions $T_1$ and $T_2$, if all ordering indicators in ${ois}_1^C$ are smaller than all those in ${ois}_2^C$, i.e.,
\[
\forall oi_1 \in {ois}_1^C, \forall oi_2 \in {ois}_2^C: oi_1 < oi_2.
\]
then $T_1$ must be ordered before $T_2$ \Changed{in the final ordering.}
\end{definition}

Figure~\ref{fig:oi} illustrates this definition with an example involving four replicas, where replicas $R_1$, $R_2$, $R_3$ are correct, and $R_4$ is Byzantine. Each replica assigns local ordering indicators to four transactions. For transactions $T_1$ and $T_4$, the correct replicas assign ${ois}_1^C = \{2, 1, 1\}$ and ${ois}_4^C = \{3, 4, 4\}$, respectively. Since \Changed{all} indicators in ${ois}_1^C$ are smaller than \Changed{those} in ${ois}_4^C$, \Changed{any protocol satisfying Ordering Linearizability—such as \Pompe{} and \FairDAG{}—must place $T_1$ before $T_4$ in the final ordering, regardless of the Byzantine replica’s input.}

\begin{figure}[t]
\begin{equation}
    \begin{split}
        R_1 &:\ \{(T_2, 1), (T_1,2), (T_4,3), (T_3,4)\} \\[-4pt]
        R_2 &:\ \{(T_1, 1), (T_3,2), (T_2,3), (T_4,4)\} \\[-4pt]
        R_3 &:\ \{(T_1, 1), (T_2,1), (T_3,3), (T_4,4)\} \\[-4pt]
        \textcolor{red}{R_4} &:\textcolor{red}{\ \{(T_1, 1), (T_2,2), (T_3,3), (T_4,4)\}\notag}
    \end{split}
\end{equation}
    
    \caption{$T_1$ will be ordered before $T_4$ if \emph{Ordering Linearizability} holds regardless of the local ordering from $R_4$.}
    \label{fig:oi}
\end{figure}

\subsection{$\gamma$-Batch-Order-Fairness}

In Section~\ref{ssec:receiveorder}, we showed that it is impossible to guarantee \emph{Receive-Order-Fairness} \Changed{in the presence of unknown Byzantine replicas.} However, a weaker yet practical variant, \emph{$\gamma$-Batch-Order-Fairness}, can be achieved by both \Themis{}~\cite{themis} and our protocol \FairDAGRL{}.

We say that a transaction $T_2$ is \textbf{dependent} on transaction $T_1$, denoted as $T_1 \rightarrow T_2$, if $T_1$ must be ordered before $T_2$ in the final ordering. Due to the presence of \emph{Condorcet cycles}, \Changed{such dependencies can form cycles among transactions.}

\begin{definition}
    A batch $S$ of transactions is \emph{cyclic dependent} if for any two transactions $T_1, T_2$ in $S$, there is a list of transactions that form a dependency path from $T_1$ to $T_2$.
\end{definition}

The final ordering can be \Changed{partitioned} into a sequence of non-overlapping \Changed{batches} $S_1, S_2, \dots$, where each $S_i$ is a \emph{maximal cyclically dependent batch},  i.e., for any $i$, $S_i \cup S_{i+1}$ is not cyclically dependent.

We say that $T_1$ is ordered \textbf{no later than} $T_2$ if $T_1$ is in the same or an earlier batch than $T_2$. And we define \emph{$\gamma$-Batch-Order-Fairness} as follows:

\begin{definition}\label{def:batchorder}
    \textbf{$\gamma$-Batch-Order-Fairness.} For any two transactions $T_1$ and $T_2$, if at least a fraction $\gamma$ of correct replicas receive $T_1$ before $T_2$, then $T_1$ must be ordered \emph{no later than} $T_2$ in the final ordering.
\end{definition}

Figure~\ref{fig:batch} presents an example satisfying \emph{$\gamma$-Batch-Order-Fairness}, where we have $\n=4$ replicas (at most $\f=1$ replica can be Byzantine) and 6 transactions. Assuming that the fairness protocol generates a final ordering that can be split into three \emph{cyclic dependent batches}, for any two transactions, \emph{$\gamma$-Batch-Order-Fairness} holds. For example, $T_0$ is received earlier than $T_1$ by correct replicas of $\gamma(\n{-}\f)=3$ and $T_0$ is ordered in a batch earlier than $T_1$.

\begin{figure}[t]
    
    \begin{tikzpicture}[yscale=0.7]
        \node at (-2.2,1.5) {$R_1: \{T_0, T_1, T_2, T_3, T_4, T_5\}$};
        \node at (-2.2,0.75) {$R_2: \{T_0, T_2, T_3, T_4, T_1, T_5\}$};
        \node at (-2.2,0) {$R_3: \{T_0, T_3, T_4, T_1, T_2, T_5\}$};
        \node at (-2.2,-0.75) {\textcolor{red}{$R_4: \{T_0, T_4, T_1, T_2, T_3, T_5\}$}};
        \node at (-2.2,-1.5) {Final Ordering: $\{T_0, T_1, T_2, T_3, T_4, T_5\}$};

        \node[minimum width=0.8cm, minimum height=0.6cm, draw, rectangle] (b1) at (2.1,1.2) {$b_1:\{T_0\}$};
        \node[minimum width=0.8cm, minimum height=0.6cm, draw, rectangle] (b2) at (2.1,0) {$b_2:\{T_1,T_2,T_3,T_4\}$};
        \node[minimum width=0.8cm, minimum height=0.6cm, draw, rectangle] (b3) at (2.1,-1.2) {$b_3:\{T_5\}$};
        \draw[->] (b1) -- (b2);
        \draw[->] (b2) -- (b3);

        \node at (2,-2.3) {(b)};
        \node at (-2,-2.3) {(a)};

    \end{tikzpicture}
    \caption{A final ordering of six transactions that satisfies $\gamma$-Batch-Order-Fairness with \Changed{$\gamma = \frac{2}{3}$, 3 correct replicas, and 1 Byzantine replica.}}
    \label{fig:batch}
\end{figure}

\subsection{DAG-based Consensus Protocols}~\label{ss:dagintro}

DAG-based BFT consensus protocols\Changed{~\cite{narwhal,dagrider,bullshark}} operate in rounds. In each round $r$, every replica proposes a block, referred to as a DAG vertex. Each vertex references multiple vertices from the previous round $r{-}1$, represented as edges that encode the causal dependencies between DAG vertices, forming a \emph{Directed Acyclic Graph (DAG)}. The \emph{causal history} of a DAG vertex includes all vertices reachable via the reference paths.

Every $k$ rounds (e.g. $k=2$ in \Tusk{}~\cite{narwhal}), a random or predetermined leader vertex is elected. These leader vertices are committed in an ascending order of round, and their causal histories are committed in a deterministic order. To ensure reliable dissemination, most DAG-based protocols employ reliable broadcast (RBC) mechanisms. With RBC and carefully designed commit rules, DAG-based protocols guarantee the following properties even in an \emph{asynchronous network} without message delay bound:

\begin{itemize}
\item \textbf{Agreement}: If a correct replica commits a vertex $v$, then all correct replicas eventually commit $v$.
\item \textbf{Total Order}: If a correct replica commits $v$ before $v'$, then every correct replica commits $v$ before $v'$.
\item \textbf{Validity}: If a correct replica broadcasts a vertex $v$, then all correct replicas eventually commit $v$.
\end{itemize}

Compared to consensus protocols with a single leader, the multi-proposer design of DAG-based protocols enables higher throughput. But this comes at the cost of higher commit latency due to the overhead of RBC and multi-round commit rules.
\section{System Model}\label{sec:model}

We consider a distributed system consisting of a set of replicas and a potentially unbounded number of clients. The system is subject to Byzantine faults and operates under either asynchronous or partially synchronous network conditions. Our model covers client behavior, replica corruption, authentication assumptions, and fairness-specific threat considerations.

\subsection{Clients}

Clients issue transactions to replicas and wait for execution results. They may behave arbitrarily with no correctness assumptions.

\subsection{Replicas}

In the system, there are a total of $\n$ replicas and an \emph{adaptive adversary} capable of corrupting up to $\f$ replicas during execution. The corrupted replicas, referred to as Byzantine or malicious replicas, may exhibit arbitrary malicious behavior.

\fullversion{
\begin{figure}[t]

\centering
\begin{tikzpicture}[
    client/.style={draw, shape=ellipse, minimum size=0.8cm, text centered, fill=blue!20},
    replica/.style={draw, cylinder, shape border rotate=90, aspect=0.3, minimum height=0.5cm, minimum width=0.5cm, text centered, fill=green!20},
    external network/.style={dashed, blue},
    internal network/.style={dashed, red},
    arrow/.style={->, thick},
    textnode/.style={draw=none, rectangle, text centered},
    >=stealth,
    font=\small,
    scale=0.5
]

\node (client1) at (0, 3
) {\includegraphics[width=1cm]{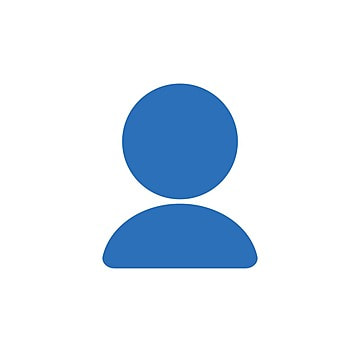}};
\node (client2) at (0, 0) {\includegraphics[width=1cm]{images/human_icon.jpg}};

\node (clientpoint1) at (1.2,3) {};
\node (clientpoint2) at (1.2,0) {};
\node (clientpoint3) at (1.2, 2.8) {};
\node (clientpoint4) at (1.2, -0.2) {};

\node (point1) at (3.5, 3) {};
\node (point2) at (3.5, 0) {};
\node (point3) at (3.5, 2.8) {};
\node (point4) at (3.5, -0.2) {};

\node[replica] (replica1) at (5, 3) {$R_1$};
\node[draw, circle, minimum size=0.2cm, draw=white, text centered] (replicapoint1) at (6, 2.5) {};
\node[replica] (replica2) at (12, 3) {$R_2$};
\node[draw, circle, minimum size=0.2cm, draw=white, text centered] (replicapoint2) at (11, 2.5) {};
\node[replica] (replica3) at (5, 0) {$R_3$};
\node[draw, circle, minimum size=0.2cm, draw=white, text centered] (replicapoint3) at (6, 0.5) {};
\node[replica] (replica4) at (12, 0) {$R_4$};
\node[draw, circle, minimum size=0.2cm, draw=white, text centered] (replicapoint4) at (11, 0.5) {};

\node[textnode] at (0, -1.5) {Clients};
\draw[external network] (1, -1) rectangle (3.6, 4);
\node[textnode] at (2.3, 1.6) {\textbf{External}};
\node[textnode] at (2.3, 1.2) {\textbf{Network}};

\node[textnode] at (8.5, -1.5) {Replicas};
\draw[internal network] (5.8, -0.5) rectangle (11.2, 3.5);
\node[textnode] at (8.4, 3) {\textbf{Internal Network}};
\node[textnode] at (8.4, 0) {Consensus Messages};

\draw[line width = 1pt] (4.2, -1) rectangle (12.8, 4);
\draw[line width = 1pt] (-0.7, -1) rectangle (0.7, 4);

\node[textnode] at (2.4, 3.5) {$T_1, T_2$};
\node[textnode] at (2.4, 2.3) {$Rsp_1, Rsp_2$};
\node[textnode] at (2.4, 0.4) {$T_3, T_4$};
\node[textnode] at (2.4, -0.5) {$Rsp_3, Rsp_4$};
\draw[->] (clientpoint1) --  (point1);
\draw[->] (clientpoint2) --  (point2);
\draw[->] (point3) --  (clientpoint3);
\draw[->] (point4) --  (clientpoint4);

\draw[<->,line width=1pt] (replicapoint1) --  (replicapoint2);
\draw[<->,line width=1pt] (replicapoint1) --  (replicapoint3);
\draw[<->,line width=1pt] (replicapoint1) --  (replicapoint4);
\draw[<->,line width=1pt] (replicapoint2) --  (replicapoint3);
\draw[<->,line width=1pt] (replicapoint2) --  (replicapoint4);
\draw[<->,line width=1pt] (replicapoint1) --  (replicapoint3);
\draw[<->,line width=1pt] (replicapoint3) --  (replicapoint4);

\end{tikzpicture}
\caption{Network topology with Clients, Replicas, and Networks. $Rsp_{i}$ represents a client response for transaction $T_i$, including the execution results of $T_i$.}
\label{fig:clients_replicas}
\end{figure}

}

Regarding transaction ordering, Byzantine replicas may reorder transactions in local orderings and ignore unfavorable local orderings from other replicas. Correct replicas are honest about local orderings in which the transactions are received.

Fairness protocols differ in resilience to Byzantine faults. Specifically, \FairDAG{} and \Pompe{} require $\n> 3\f$; \Themis{} requires $\n>\frac{(2\gamma+2)\f}{2\gamma-1}$; and \FairDAGRL{} requires $\n>\frac{(2\gamma+1)\f}{2\gamma-1}$, $\frac{1}{2} < \gamma \le 1$.

\subsection{Authentication}

We assume \emph{authenticated communication}, where Byzantine replicas cannot forge messages from correct replicas. Authentication is enforced through \emph{Public Key Infrastructure (PKI)}~\cite{cryptobook}. 

For integrity verification, each transaction $T_i$ is associated with a digest $d_i$, computed using a \emph{secure collision-resistant cryptographic hash function}~\cite{cryptobook}.

\subsection{Network}

\fullversion{
As illustrated in Figure~\ref{fig:clients_replicas}, we distinguish between two types of networks: (1) \emph{external communication} between clients and replicas, and (2) \emph{internal communication} among replicas.}

\begin{flushleft}
\textbf{External network (client–replica).}
We make no assumptions about synchrony but assume the external network is \emph{non-adversarial}. This assumption is necessary for preserving fairness, as an adversarial external network could arbitrarily control the order in which replicas receive transactions and then the final ordering.
\end{flushleft}

\begin{flushleft}
\textbf{Internal network (replica–replica).}
The internal network may operate either \emph{asynchronously} or \emph{partially synchronously}:
\end{flushleft}

\begin{itemize}
\item \textbf{Asynchronous network:} Messages are never lost and are eventually delivered, but there is no bound on the message delays.
\item \textbf{Partially synchronous network:} There exists an unknown \emph{Global Stabilization Time (GST)} after which the message delays are bounded by a known constant $\Delta$, i.e., any message sent at time $t$ will be delivered by $\max(t, \text{GST}) + \Delta$.
\end{itemize}

\section{Overview}\label{sec:overviewfairdag}
\begin{figure}[t]
\centering
\begin{tikzpicture}[
    client/.style={draw, shape=ellipse, minimum size=0.8cm, text centered, fill=blue!20},
    replica/.style={draw, cylinder, shape border rotate=90, aspect=0.3, minimum height=0.5cm, minimum width=0.5cm, text centered, fill=green!20},
    external network/.style={dashed, blue},
    internal network/.style={dashed, red},
    arrow/.style={->, thick},
    textnode/.style={draw=none, rectangle, text centered},
    >=stealth,
    yscale=0.9,
    font=\small
]

\node[replica] (replica1) at (0, 8.7) {$R_1$};
\node (broadcast1) at (0.7, 8.7) {\includegraphics[width=1cm]{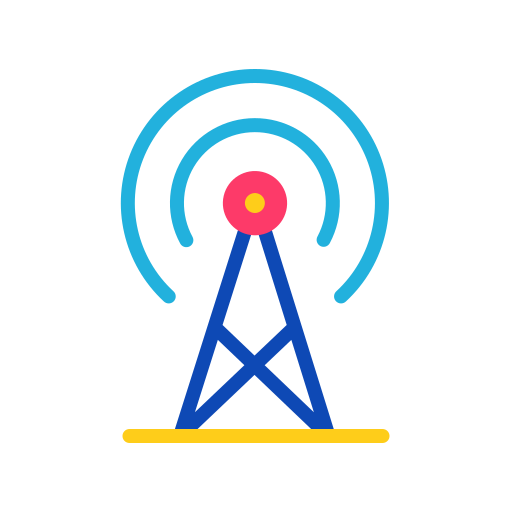}};
\node[font=\tiny,rectangle, draw] (lo11) at (-0.2, 8) {$v_{1,1}: LO_{1,1}$};
\node[font=\tiny,rectangle, draw] (lo12) at (0.5, 7.5) {$v_{1,2}: LO_{1,2}$};
\node[font=\tiny,rectangle, draw] (lo13) at (1.2, 8) {$v_{1,3}: LO_{1,3}$};

\node[replica] (replica2) at (4, 8.7) {$R_2$};
\node (broadcast2) at (3.3, 8.7) {\includegraphics[width=1cm]{images/broadcast.png}};
\node[font=\tiny,rectangle, draw] (lo21) at (3.3, 8) {$v_{2,1}: LO_{2,1}$};
\node[font=\tiny,rectangle, draw] (lo22) at (4, 7.5) {$v_{2,2}: LO_{2,2}$};
\node[font=\tiny,rectangle, draw] (lo23) at (4.7, 8) {$v_{2,3}: LO_{2,3}$};

\node[replica] (replica3) at (0, 6.7) {$R_3$};
\node (broadcast2) at (0.7, 6.7) {\includegraphics[width=1cm]{images/broadcast.png}};
\node[font=\tiny,rectangle, draw] (lo31) at (-0.2, 6) {$LO_{3,1}:\{d_3\}$};
\node[font=\tiny,rectangle, draw] (lo32) at (0.5, 5.5) {$LO_{3,2}:\{d_1,d_4\}$};
\node[font=\tiny,rectangle, draw] (lo33) at (1.2, 6) {$LO_{3,3}:\{d_2\}$};

\node[replica] (replica4) at (4, 6.7) {$R_4$};
\node (broadcast2) at (3.3, 6.7) {\includegraphics[width=1cm]{images/broadcast.png}};
\node[font=\tiny,rectangle, draw] (lo41) at (3.3, 6) {$LO_{4,1}:\{d_1\}$};
\node[font=\tiny,rectangle, draw] (lo42) at (4, 5.5) {$LO_{4,2}:\{d_4,d_2\}$};
\node[font=\tiny,rectangle, draw] (lo43) at (4.7, 6) {$LO_{4,3}:\{d_3\}$};

\node at (5, 7) {(a)};
\node at (4, 10) {(b)};
\node at (5, 12) {(c)};
\node at (5, 13.9) {(d)};
\node at (5, 15.0) {(e)};
\node at (5, 16.6) {(f)};

\draw[->, gray, line width=8pt] (2.2, 9) -- (2.2, 10.5);
\node (receiver) at (3, 10) {\includegraphics[width=0.8cm]{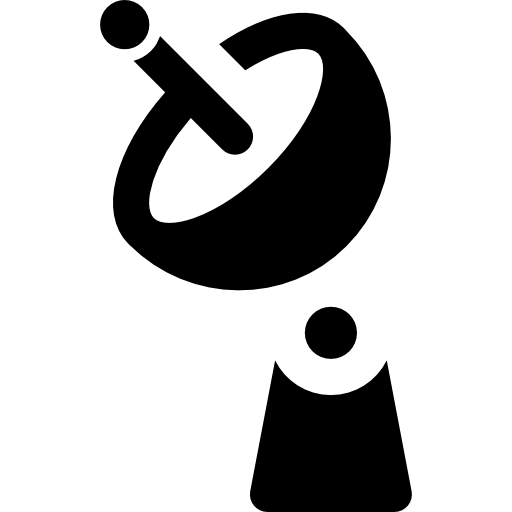}};

\node[font=\tiny,rectangle, draw, fill=cyan!20] (dag11) at (0, 13) {$LO_{1,1}$};
\node[font=\tiny,rectangle, draw, fill=green!20] (dag21) at (0, 12.3) {$LO_{2,1}$};
\node[font=\tiny,rectangle, draw, fill=cyan!20] (dag31) at (0, 11.6) {$LO_{3,1}$};
\node[font=\tiny,rectangle, draw, fill=cyan!20] (dag41) at (0, 10.9) {$LO_{4,1}$};

\node[font=\tiny,rectangle, draw, fill=cyan!20] (dag12) at (2, 13) {$LO_{1,2}$};
\node[font=\tiny,rectangle, draw, fill=cyan!20] (dag22) at (2, 12.3) {$LO_{2,2}$};
\node[font=\tiny,rectangle, draw, fill=cyan!20] (dag32) at (2, 11.6) {$LO_{3,2}$};

\node[font=\tiny,rectangle, draw] (dag13) at (4, 13) {$LO_{1,3}$};
\node[font=\tiny,rectangle, draw] (dag23) at (4, 12.3) {$LO_{2,3}$};
\node[font=\tiny,rectangle, draw, fill=cyan!20] (dag33) at (4, 11.6) {$LO_{3,3}$};

\draw[->, line width=0.6pt] (dag12) -- (dag11);
\draw[->, line width=0.6pt] (dag12.west) -- (dag21.east);
\draw[->, line width=0.6pt] (dag12.west) -- (dag31.east);

\draw[->, line width=0.6pt] (dag22) -- (dag21);
\draw[->, line width=0.6pt] (dag22.west) -- (dag31.east);
\draw[->, line width=0.6pt] (dag22.west) -- (dag11.east);

\draw[->, line width=0.6pt] (dag32) -- (dag31);
\draw[->, line width=0.6pt] (dag32.west) -- (dag21.east);
\draw[->, line width=0.6pt] (dag32.west) -- (dag11.east);

\draw[->, line width=0.6pt] (dag13) -- (dag12);
\draw[->, line width=0.6pt] (dag13.west) -- (dag22.east);
\draw[->, line width=0.6pt] (dag13.west) -- (dag32.east);

\draw[->, line width=0.6pt] (dag23) -- (dag22);
\draw[->, line width=0.6pt] (dag23.west) -- (dag12.east);

\draw[->, line width=0.6pt] (dag33) -- (dag32);
\draw[->, line width=0.6pt] (dag33.west) -- (dag22.east);

\draw[->, line width=0.6pt, blue!30] (dag33.west) -- (dag41.east);

\draw[->, gray, line width=6pt] (2.3, 13.7) -- (2.3, 14.5);

\node[fill=green!20, draw, rectangle, minimum width=2pt, minimum height=1pt] at (0.9, 14.2) {};  
\node[fill=green!20, draw, rectangle, minimum width=2pt, minimum height=1pt] at (1.2, 14.2) {};  
\node[fill=green!20, draw, rectangle, minimum width=2pt, minimum height=1pt] at (1.5, 14.2) {};  

\node[fill=cyan!20, draw, rectangle, minimum width=2pt, minimum height=1pt] at (0.9, 13.8) {};  
\node[fill=cyan!20, draw, rectangle, minimum width=2pt, minimum height=1pt] at (1.2, 13.8) {};  
\node[fill=cyan!20, draw, rectangle, minimum width=2pt, minimum height=1pt] at (1.5, 13.8) {};  

\node[font=\tiny] at (-0.2,14.2) {After $L_1$ is committed:};
\node[font=\tiny] at (-0.2,13.8) {After $L_3$ is committed:};

\node[font=\tiny] at (1.9,14.2) {$A_1$};
\node[font=\tiny] at (1.9,13.8) {$A_3$};

\node at (2.2, 16.7) {$T_1 \rightarrow T_2 \rightarrow T_3 \rightarrow T_4$};

\node[minimum width=0.4cm,draw,circle,font=\large, scale=0.5] (b2) at (2.8,15.4) {$d_1$};
\node[minimum width=0.4cm,draw,circle,font=\large, scale=0.5] (c2) at (3.3,15.1) {$d_2$};
\node[minimum width=0.4cm,draw,circle,font=\large, scale=0.5] (d2) at (3.8,14.8) {$d_3$};
\node[minimum width=0.4cm,draw,circle,font=\large, scale=0.5] (e2) at (3.8,15.4) {$d_4$};
\draw[->] (b2) -- (c2);
\draw[->] (c2) -- (d2);
\draw[->] (d2) -- (e2);
\draw[->] (e2) -- (c2);

\node (processor) at (0.8, 15.1) {\includegraphics[width=0.8cm]{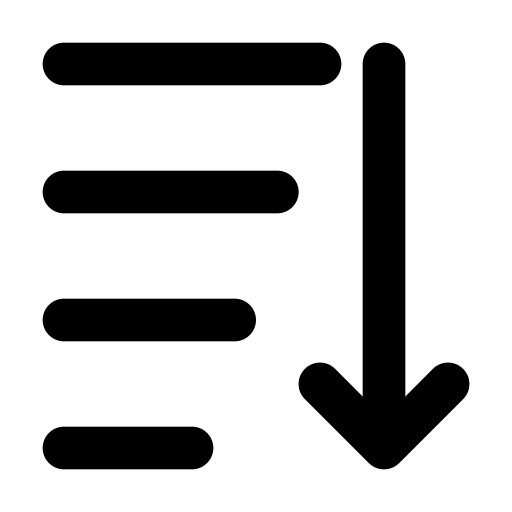}};

\node[font=\tiny] at (1.6, 15.4) {$d_1.AOI$};
\node[font=\tiny] at (1.6, 15.2) {$d_2.AOI$};
\node[font=\tiny] at (1.6, 15) {$d_3.AOI$};
\node[font=\tiny] at (1.6, 14.8) {$d_4.AOI$};

\draw[->, gray, line width=6pt] (2.2, 15.5) -- (2.2, 16.3);

\draw[dashed, blue] (-1, 5.2) rectangle (5.5, 13.5);
\node at (0, 10) {\textcolor{blue}{\textbf{DAG Layer}}};

\draw[dashed, orange] (-1, 14.5) rectangle (5.5, 17);
\node at (0, 16.1) {\textcolor{orange}{\textbf{Fairness Layer}}};
\end{tikzpicture}
\caption{Architecture of FairDAG: (a) Replicas reliably broadcast blocks containing their local ordering \Changed{fragments}. (b) Each replica receives blocks delivered through Reliable Broadcast. \Changed{(c) Each replica forms a local view of the DAG using received blocks and reference links, where different colors represent the subdags of different committed leader vertices. }(d) Local orderings in subdag $A_r$ are used as input of the fairness layer after $L_r$ is committed. \Changed{(e) Finalize transaction ordering using absolute ordering mechanism (left) or relative ordering mechanism (right), based on the committed local orderings. (f)} A final transaction ordering is generated.}
\label{fig:overview}
\end{figure}


Previous fairness protocols, such as \Pompe{}~\cite{pompe} and \Themis{}~\cite{themis}, \Changed{are built atop} leader-based consensus. Although the final ordering is derived from local orderings collected from a majority of the replicas, existing fairness protocols have these problems: (1) Byzantine leaders manipulate orderings and compromise fairness by selectively picking up the local orderings. (2) Byzantine or slow leaders can become system performance bottlenecks. To address these problems, \FDAG{} runs fairness protocols on top of multi-proposer DAG-based consensus with two layers:

\begin{itemize}
    \item \emph{DAG Layer} handles the dissemination (Figure~\ref{fig:overview}(a,b)) and commitment (Figure~\ref{fig:overview}(c)) of local orderings.
    \item \emph{Fairness Layer} takes these committed local orderings (Figure~\ref{fig:overview}(d)) as input to a transaction ordering algorithm (Figure~\ref{fig:overview}(e)) to produce a final ordering (Figure~\ref{fig:overview}(f)).
\end{itemize}

\subsection{DAG Layer}\label{ssec:daglayer}

The DAG layer adopts existing DAG-based consensus protocols—such as Tusk~\cite{narwhal}, DAG-Rider~\cite{dagrider}, and Bullshark~\cite{bullshark}{—with minimal modifications to support fairness guarantees.

\begin{flushleft}
    \textbf{Directed-Acylic-Graph}
\end{flushleft}

DAG-based protocols proceed \Changed{in rounds where each replica concurrently} proposes one vertex per round $r$. \Changed{Each vertex references previous vertices, forming a DAG, where references are \textbf{edges}. Specifically, a vertex in round $r$ references at least $\n{-}\f$ vertices from round $r-1$ via \textbf{strong edges}, and optionally up to $\f$ vertices from earlier rounds via \textbf{weak edges}.}

Let $v_{i,r}$ denote the vertex proposed by replica $R_i$ in round $r$. A valid vertex satisfies:

\begin{itemize}
    \item \textbf{$strong\_edges$}: includes at least $\n-\f$ \emph{strong edges},
    \item \textbf{$weak\_edges$}: includes \Changed{up to} $\f$ \emph{weak edges}.
\end{itemize}

To guarantee the fairness properties in the presence of adversary, we extend existing DAG-based protocols by applying new rules of forming DAG vertices:
\begin{enumerate}
    \item In \FDAG{}, replicas include weak edges in vertices to incorporate local orderings from slower replicas, ensuring that all correct local orderings eventually get committed.
    \item To preserve the integrity of a replica's local ordering, each vertex $v_{i,r}$ must include a strong edge to its vertex in the previous round  $v_{i,r-1}$ (for $r > 0$).
\end{enumerate}
Figure~\ref{fig:overview}(c) shows an DAG example where black arrows represent strong edges, and blue arrows represent weak edges. 

\begin{flushleft}
    \textbf{\Changed{Reliably broadcasting} a vertex}
\end{flushleft}

To address issue (1), in \FDAG{}, instead of relying on a leader to collect local orderings, each replica $R_i$ autonomously reliably broadcast~\cite{narwhal,bracha1987asynchronous,cachin2005asynchronous,guerraoui2019scalable} its  local ordering $LO_i$.
As Figure~\ref{fig:overview}(a) shows, \FDAG{} vertex $v_{i,r}$ contains $r$-th slice of the replica $R_i$’s \emph{local ordering} ($LO_{i,r}$)—a sequence reflecting the order in which transactions were received by replica $R_i$. Each local ordering is represented as a sequence of transaction digests paired with monotonically increasing ordering indicators.

\begin{flushleft}
    \textbf{Constructing a DAG}
\end{flushleft}

As illustrated in \Changed{Figure~\ref{fig:overview} (a,b), the DAG layer of \FDAG{} enables replicas to reliably broadcast and deliver vertices from one another. This multi-proposer design alleviates the performance bottlenecks associated with a single leader, thereby addressing issue (2).

Using the delivered vertices and the references within them, each replica constructs its local view of the DAG,} as shown in Figure~\ref{fig:overview}(c). As mentioned in Section~\ref{ss:dagintro}, the DAG protocols guarantee that all correct replicas eventually have consistent DAG views.

\begin{flushleft}
    \textbf{Committing DAG vertices}
\end{flushleft}

The commit rules in DAG-based consensus protocols—leveraging the causal dependencies between vertices—further mitigate leader manipulation in selecting local orderings, addressing issue (1).

DAG protocols \Changed{group rounds} into \emph{waves}, \Changed{each containing one or more leader vertices. A leader vertex $L_{r_i}$ in round $r_i$ is \emph{committed} once specific conditions are met.} Let $C_{r_i}$ denote its \emph{causal history}—\Changed{the set of vertices reachable from $L_{r_i}$,} including $L_{r_i}$ itself. DAG protocols ensure that all correct replicas commit leader vertices in a consistent, round-increasing order $(L_{r_1}, L_{r_2}, \dots)$ such that $C_{r_i} \subset C_{r_{i+1}}$ for all $i$.

A vertex $v$ is said to be in the subdag of leader $L_{r_i}$ if $v \in C_{r_i}$ and $v \notin C_{r_j}$ for all $j < i$. Let $A_{r_i}$ denote the vertices in the subdag of $L_{r_i}$. The committed DAG can thus be partitioned into non-overlapping subdags $(A_{r_1}, A_{r_2}, \dots)$, each corresponding to a committed leader.

For example, in Figure~\ref{fig:overview}(c), $L_1$ is the green vertex $v_{2,1}$ and $L_3$ is the blue vertex $v_{3,3}$. Then $A_1$ consists of the green vertices, and $A_3$ consists of the blue ones.

As shown in Figure~\ref{fig:overview}(d), each time a leader $L_r$ is committed, the fairness layer \Changed{processes $A_r$ to compute} the final ordering.

The DAG layer satisfies a necessary condition to guarantee that final ordering aligns with the preferences of the majority of correct replicas, as each subdag $A_r$ aggregates local orderings from at least $\n{-}\f$ replicas—guaranteeing the inclusion of at least $\n{-}2\f$ correct replicas. (Note: $\n{-}2\f \ge \frac{\n{-}\f}{2}$ holds for all protocols.)

\subsection{Fairness Layer}

Taking local orderings within the committed subdags as input, the fairness layer runs a transaction ordering algorithm to calculate the final transaction ordering.

In \textbf{absolute} fairness protocols (Figure~\ref{fig:overview}(e), left), each transaction gets an \emph{assigned ordering indicator (AOI) based on the input local orderings.} The final ordering (Figure~\ref{fig:overview}(f)) is then derived, sorting transactions by their \emph{AOI}. \FairDAG{} is an \emph{absolute} fairness protocols and satisfies the \emph{Ordering Linearizability} property.

In \textbf{relative} fairness protocols (Figure~\ref{fig:overview}(e), right) construct a dependency graph of transaction nodes, where edges encode pairwise ordering dependencies derived from the input local orderings. The final ordering is a \emph{Hamiltonian path} within the graph (Figure~\ref{fig:overview}(f)). \FairDAGRL{} is an \emph{relative} fairness protocol and satisfies the \emph{$\gamma$-Batch-Order-Fairness} property.

\Changed{\section{FairDAG-AB}\label{sec:design}
}
\FairDAG{} is an \textbf{absolute} fairness protocol. In this section, we demonstrate how \FairDAG{} calculates the \emph{assigned ordering indicator (AOI)} for each transaction and orders the transactions based on the AOI values.

\subsection{Transaction Dissemination}

\Changed{To prevent leader-driven manipulation}, clients broadcast transactions to all replicas. This prevents adversarial replicas from censoring transactions to influence their position in final ordering.

Upon receiving a client transaction $T$, replica $R$ gives $T$ a monotonically increasing ordering indicator $oi$ from its local timer, appends digest $d$ of $T$ and $oi$ to its lists $dgs$ and $ois$ (Figure~\ref{fig:alg-ab}, Line~\ref{alg:dgs}), respectively, which are in-memory variables storing transaction digests and ordering indicators of pending transactions (Lines~\ref{alg:recv1}-\ref{alg:recv2}).

\subsection{FairDAG-AB Vertex}\label{ss:abvertex}

Built atop a DAG-based consensus protocol, each \FairDAG{} replica constructs and reliably broadcasts a DAG vertex when protocol conditions are met. In addition to DAG-specific metadata (as described in Section~\ref{ssec:daglayer}), a \FairDAG{} vertex includes \Changed{(and clears)} the replica’s local ordering of pending transactions it received after it broadcast the last vertex, encoded in $dgs$, $ois$ (Lines~\ref{alg:propose1}-\ref{alg:propose2}).

\subsection{Managing and Assigning Ordering Indicators}\label{ssec:assigned_oi}

\FairDAG{} leverages local orderings within committed DAG vertices to determine the final transaction ordering. However, a delay exists between receiving and committing the DAG vertices. To efficiently manage and utilize local orderings from both committed and uncommitted vertices, \FairDAG{} maintains an \emph{Ordering Indicator Manager} ($OIM(d)$) for each transaction digest $d$, which tracks the replica’s local view of the DAG and ordering indicators within the vertices.

\begin{figure}[t!]
    \begin{myprotocol}
    \STATE \textbf{\Changed{State Variables (per replica)}}\label{fig:ab-state0}
        \STATE $\twat{}\GETS \{\}$
        \STATE $\hts{}[1..\n] \GETS 0$\label{alg:hoi}
        \STATE $dgs, ois \GETS []$ \label{alg:dgs}
        \STATE $current\_round, replica\_id$\label{fig:ab-state1}
        
        \SPACE

    \STATE \textbf{\Changed{\textbf{Ordering Indicator Manager ($OIM$)\label{alg:oim}}}}
        \STATE \Changed{$seen\_{ois}[1..\n] \GETS \infty,\; committed\_{ois}[1..\n] \GETS \infty$}
        \STATE \Changed{$LPAOI\GETS\infty,\;AOI\GETS\infty$\label{alg:oim1}}
        
        \SPACE

        \TITLE{Client Thread}{processing client transactions}
        \EVENT{\Changed{On receive} transaction $T$}\label{alg:recv1}
            \IF{$T$ is valid}
                \STATE $oi \GETS \lotm{}$
                \STATE \Changed{$dgs.append(T.digest)$; $ois.append(oi)$}\label{alg:recv2}
            \ENDIF
        \ENDEVENT
        \SPACE

    \TITLE{DAG Layer Thread}{constructing the DAG}
    \EVENT{\Changed{On propose DAG vertex}\label{alg:propose1}}
        \STATE $v \GETS DAGVertex(replica\_id, current\_round)$
        \STATE Add pending local ordering $dgs$ and $ois$ into $v$
        \STATE Clear $dgs$ and $ois$
        \STATE Reliably broadcast $v$\label{alg:propose2}
    \ENDEVENT
        \SPACE

    \EVENT{On deliver $v_{i,r}$\label{alg:receive}}
        \STATE $\hts[i] \GETS$ the highest $oi$ in $v_{i,r}.ois$
        \FOR{$(d, oi) \in (v_{i,r}.dgs, v_{i,r}.ois)$}
            \STATE $OIM(d).seen\_ois[i] \GETS oi$\label{alg:receive1}
        \ENDFOR
    \ENDEVENT
        \SPACE

    \EVENT{\Changed{On commit $L_r$}\label{alg:commit}}
        \STATE Send $A_r$ to Fairness Layer\label{alg:commit1}
    \ENDEVENT
        \SPACE

        \TITLE{Fairness Layer Thread}{Ordering transactions}
        \STATE \COMMENT{Update $committed\_ois$}
        \EVENT{\Changed{On receive $A_r$\label{alg:commitoi}}}
        \FOR{$v \in A_r$, $(d, oi) \in (v.dgs, v.ois)$}
            \STATE $\Changed{OIM(d).committed\_ois[v.replica\_id] \GETS oi}$\label{alg:commitoi1}
        \ENDFOR

        \STATE $LPAOI_{\min} \GETS \infty$
        \FOR{$d$ \Changed{such that $OIM(d) = \infty$}\label{alg:calculate}}
            \IF{$OIM(d)$ has at least $\n-\f$ $committed\_ois$}
                \STATE \COMMENT{Calculate $AOI$}
                \STATE \Changed{$OIM(d).AOI \GETS (\f{+}1)$-th smallest in $committed\_ois$}
                \STATE \Changed{Add $d$ to $\twat$\label{alg:calculate1}}
            \ELSE\label{alg:lpaoi}
                \STATE \COMMENT{Calculate $LPAOI$}
               \FOR{\Changed{$i = 1$ to $\n$}}
                    \STATE $lp\_ois[i] \GETS \min(d.seen\_ois[i], \hts[i])$
                \ENDFOR
                \STATE \Changed{$d.LPAOI \GETS (\f{+}1)$-th smallest in $lp\_ois$}
                \STATE \COMMENT{Track the minimal LPAOI of all transactions without an AOI}
                \STATE $LPAOI_{\min} \GETS \min(LPAOI_{\min}, d.LPAOI)$\label{alg:lpaoi1}
            \ENDIF
        \ENDFOR

        \STATE sort $\twat$ sorted by $AOI$\label{line:sort}
        \STATE execute all transactions with an $AOI$ lower than $LPAOI_{min}$\label{line:execute}
    \ENDEVENT
    \end{myprotocol}
    \caption{FairDAG-AB Algorithm.}
    \label{fig:alg-ab}
\end{figure}
\begin{figure}[t]
    \begin{tikzpicture}[yscale=0.5, xscale=0.6]
    \tiny
    \definecolor{softblue}{RGB}{173, 200, 230}
    \definecolor{softpink}{RGB}{255, 182, 193}
    \definecolor{softgreen}{RGB}{144, 238, 80}
    \definecolor{softyellow}{RGB}{255, 215, 0}
        
        \draw[softyellow, thick, line width=0.3mm] (0,0) rectangle (2,1.5);
        \draw[softyellow, line width=0.3mm] (4,0) rectangle (6,1.5);
        \draw[softgreen, line width=0.3mm] (8,0) rectangle (10,1.5);
        \draw[black, line width=0.3mm] (12,0) rectangle (14,1.5);
        \draw[softyellow, line width=0.3mm] (0,2) rectangle (2,3.5);
        \draw[softgreen, line width=0.3mm] (4,2) rectangle (6,3.5);
        \draw[softgreen, line width=0.3mm] (8,2) rectangle (10,3.5);
        \draw[softgreen, line width=0.3mm] (12,2) rectangle (14,3.5);
        \draw[softyellow, line width=0.3mm] (0,4) rectangle (2,5.5);
        \draw[softgreen, line width=0.3mm] (4,4) rectangle (6,5.5);
        \draw[softgreen, line width=0.3mm] (8,4) rectangle (10,5.5);
        \draw[black, line width=0.3mm] (12,4) rectangle (14,5.5);
        \draw[softgreen, line width=0.3mm] (0,6) rectangle (2,7.5);

        \node[align=center,anchor=north] at (1,1.5) {$v_{4,1}$\\\textcolor{red}{$(\digest{1}, 1)$}\\\textcolor{red}{$(\digest{2}, 2)$}};
        \node[align=center,anchor=north] at (1,3.5) {$v_{3,1}$\\$(\digest{2}, 1)$\\$(\digest{1}, 2)$};
        \node[align=center,anchor=north] at (1,5.5) {$v_{2,1}$\\$(\digest{1}, 1)$};
        \node[align=center,anchor=north] at (1,7.5) {$v_{1,1}$\\$(\digest{1}, 1)$};
        \node[align=center,anchor=north] at (5,1.5) {$v_{4,2}$\\\textcolor{red}{$(\digest{3}, 3)$}};
        \node[align=center,anchor=north] at (5,3.5) {$v_{3,2}$\\$(\digest{4}, 3)$};
        \node[align=center,anchor=north] at (5,5.5) {$v_{2,2}$\\$(\digest{2}, 2)$\\$(\digest{4}, 3)$};
        \node[align=center,anchor=north] at (9,1.5) {$v_{4,3}$\\$(\digest{5}, 4)$};
        \node[align=center,anchor=north] at (9,3.5) {$v_{3,3}$\\$(\digest{3}, 4)$};
        \node[align=center,anchor=north] at (9,5.5) {$v_{2,3}$\\$(\digest{3}, 4)$};
        \node[align=center,anchor=north] at (13,1.5) {$v_{4,4}$\\$(\digest{4}, 5)$};
        \node[align=center,anchor=north] at (13,3.5) {$v_{3,4}$\\$(\digest{6}, 5)$};
        \node[align=center,anchor=north] at (13,5.5) {$v_{2,4}$\\$(\digest{5}, 5)$};

        \draw[->,line width = 0.3mm] (4,0.75) -- (2,0.75);
        \draw[->,line width = 0.3mm] (4,0.75) -- (2,2.75);
        \draw[->,line width = 0.3mm] (4,0.75) -- (2,4.75);
        \draw[->,line width = 0.3mm] (4,2.75) -- (2,0.75);
        \draw[->,line width = 0.3mm] (4,2.75) -- (2,2.75);
        \draw[->,line width = 0.3mm] (4,2.75) -- (2,4.75);
        \draw[->,line width = 0.3mm] (4,4.75) -- (2,0.75);
        \draw[->,line width = 0.3mm] (4,4.75) -- (2,2.75);
        \draw[->,line width = 0.3mm] (4,4.75) -- (2,4.75);

        \draw[->,line width = 0.3mm] (8,0.75) -- (6,0.75);
        \draw[->,line width = 0.3mm] (8,0.75) -- (6,2.75);
        \draw[->,line width = 0.3mm] (8,0.75) -- (6,4.75);
        \draw[->,line width = 0.3mm] (8,2.75) -- (6,0.75);
        \draw[->,line width = 0.3mm] (8,2.75) -- (6,2.75);
        \draw[->,line width = 0.3mm] (8,2.75) -- (6,4.75);
        \draw[->, line width = 0.3mm, softblue] (8,4.75) -- (2,6.75);
        \draw[->,line width = 0.3mm] (8,4.75) -- (6,0.75);
        \draw[->,line width = 0.3mm] (8,4.75) -- (6,2.75);
        \draw[->,line width = 0.3mm] (8,4.75) -- (6,4.75);

        \draw[->,line width = 0.3mm] (12,0.75) -- (10,0.75);
        \draw[->,line width = 0.3mm] (12,0.75) -- (10,2.75);
        \draw[->,line width = 0.3mm] (12,0.75) -- (10,4.75);
        \draw[->,line width = 0.3mm] (12,2.75) -- (10,0.75);
        \draw[->,line width = 0.3mm] (12,2.75) -- (10,2.75);
        \draw[->,line width = 0.3mm] (12,2.75) -- (10,4.75);
        \draw[->,line width = 0.3mm] (12,4.75) -- (10,0.75);
        \draw[->,line width = 0.3mm] (12,4.75) -- (10,2.75);
        \draw[->,line width = 0.3mm] (12,4.75) -- (10,4.75);

    \end{tikzpicture}
    \caption{Example of calculating AOI and LPAOI.}
    \label{fig:assignedoi}
\end{figure}

\Changed{We say that a replica $R$ has \textbf{\emph{seen}} an ordering indicator $oi$ if $oi$ is in a vertex in $R$'s local DAG view. $R$ has \textbf{\emph{committed}} $oi$ if $oi$ is in a committed vertex. $OIM(d)$ contains the following information (Lines~\ref{alg:oim}-\ref{alg:oim1})}:

\begin{itemize}
    \item $seen\_{ois}$: ordering indicators of $d$ that $R$ has seen.
    \item $committed\_{ois}$: ordering indicators of $d$ that $R$ has committed.
    \item $LPAOI$: the lowest possible value of the assigned ordering indicator of $d$.
    \item $AOI$: the value of the assigned ordering indicator of $d$.
\end{itemize}

The final ordering is determined using $AOI$ values derived from $committed\_ois$. $LPAOI$, computed from $seen\_ois$, helps determine when it is safe to decide the position of a transaction in the final ordering (see Section~\ref{ss:global_ordering}). The $seen\_ois$ and $committed\_ois$ are indexed from $1$ to $\n$ and each item is initialized to $\infty$. For example, if $OIM(d).seen\_ois[2] = 3$, it means that $R$ has seen an ordering indicator $3$ for the transaction with digest $d$ from replica $R_2$.



\begin{flushleft}
    \textbf{Managing Ordering Indicators}
\end{flushleft}

\Changed{Besides the $OIM$ of each transaction digest $d$, each \FairDAG{} replica $R$ also maintains $highest\_ois$: a vector of the highest ordering indicators received from each replica, initialized to be 0 and indexed from 1 to $\n$ (Line~\ref{alg:hoi}). This vector is essential for checking whether it is safe to determine the position of a transaction in final ordering (See details in Section~\ref{ss:global_ordering}).}

Upon a DAG vertex $v_{i,r}$ is delivered and added to replica $R$'s local DAG view, $R$ updates $highest\_ois[i]$ to the highest ordering indicator in $v_{i,r}$. For each digest $d$ in $v_{i,r}$, $R$ updates $OIM(d).seen\_ois[i]$ (Lines~\ref{alg:receive}-\ref{alg:receive1}).

Upon a DAG leader vertex $L_r$ is committed, its subdag $A_r$ is input to the fairness layer (Lines~\ref{alg:commit}-\ref{alg:commit1}). And $R$ updates $committed\_ois$ for the digests in $A_r$ accordingly (Lines~\ref{alg:commitoi}-\ref{alg:commitoi1}).

\begin{flushleft}
    \textbf{Calculating Assigned Ordering Indicator}
\end{flushleft}

If $OIM(d)$ has at least $\n{-}\f$ non-$\infty$ $committed\_oi$ values $R$ calculates $OIM(d).AOI$. The $AOI$ is the $(\f{+}1)$-th smallest value of the $committed\_ois$, which definitely falls within the range of ordering indicators from correct replicas. Then, $d$ is added to $\twat{}$, which is set of transactions digests with valid $AOI$ values (Lines~\ref{alg:calculate}-\ref{alg:calculate1}). \emph{Note: The value of $AOI$ is immutable once calculated, even if more ordering indicators of $d$ are committed.}

\begin{example}
\Changed{Figure~\ref{fig:assignedoi} illustrates how to calculate $AOI$ values. Suppose $v_{4,2}$ and $v_{3,4}$ are committed leader vertices $L_2$ and $L_4$, with yellow $A_2$ and green $A_4$, respectively. In $A_2$, digest $d_1$ has committed ordering indicators ${\infty, 1, 2, 1}$. Since this includes at least $\n{-}\f$ valid values, its $AOI$ is the $(\f{+}1)$-th lowest, which is 1 (Lines~\ref{alg:calculate}-\ref{alg:calculate1}). Although an additional $seen\_ois$ of 1 appears in $v_{1,1} \in A_4$, it is ignored since $AOI$ is immutable once calculated after $L_2$ is committed. In $A_4$, digest $d_2$ has committed ordering indicators ${\infty, 1, 2, 3}$, yielding an $AOI$ of 2. Similarly, $d_3$ receives an $AOI$ of 3. Other digests in $A_4$ lack sufficient $committed\_ois$ data for $AOI$.}
\end{example}

\subsection{Global Ordering}~\label{ss:global_ordering}

Due to the randomness in message arrival times, DAG vertices with smaller local ordering indicators may be committed later than vertices from other replicas with higher ordering indicators. As a result, a transaction $d_1$ may get a smaller $AOI$ than $d_2$, even if $d_1$ obtains its $AOI$ in a later round than $d_2$. To ensure strict ordering of transactions based on their $AOI$ values—which is essential for achieving \emph{Ordering Linearizability}—it is necessary to guarantee that no other transaction could get a lower $AOI$ before determining the position of a transaction in the final ordering. To enforce it, for each transaction without an $AOI$, we calculate its $LPAOI$, the lowest AOI value it could possibly get. And we track $LPAOI_{min}$, the minimal $LPAOI$ of transactions without an $AOI$.

\begin{flushleft}
    \textbf{Calculating LPAOI}
\end{flushleft}

For each transaction digest $d$ without a valid $AOI$, replica $R$ constructs a vector $lp\_ois$, where each entry is computed as $lp\_ois[i] = \min(OIM(d).seen\_ois[i], highest\_ois[i])$. The $LPAOI$ of $d$ is then defined as the $(\f{+}1)$-th smallest value in $lp\_ois$. And $LPAOI_{\min}$ is the minimal value across all $LPAOI$ values  (Lines~\ref{alg:lpaoi}-\ref{alg:lpaoi1}). 

$LPAOI_{\min}$ serves as a threshold: only transactions with $AOI < LPAOI_{\min}$ can determine its position in final ordering. The fairness layer sorts all digests in $\twat{}$ by $AOI$, executes all transactions with an $AOI$ lower than the threshold. The ordered digests are then removed from $\twat{}$ (Lines~\ref{line:sort}-\ref{line:execute}).

\begin{example}
    As shown in Figure~\ref{fig:assignedoi}, transaction digests $d_1$, $d_2$, and $d_3$ have $AOI$ values of $1$, $2$, and $4$, respectively. \Changed{Given $highest\_ois = (2, 6, 6, 6)$ and $OIM(d_4).seen\_ois = (\infty, 3, 3, 5)$, digest $d_4$ derives its $lp\_ois = (2, 3, 3, 5)$ and obtains $LPAOI = 3$. Similarly, $d_5$ and $d_6$ have $LPAOI$ values of $4$ and $5$. $LPAOI_{\min}$ is then $3$, (Lines~\ref{alg:lpaoi}-\ref{alg:lpaoi1})} so $d_1$ and $d_2$ can be ordered and executed. In contrast, $d_3$ cannot be executed since its $AOI$ exceeds $LPAOI_{\min}$, implying that $d_4$ could possibly get a lower $AOI$ than $d_3$.
\end{example}

\section{FairDAG-RL}~\label{sec:design2}

\FairDAGRL{} is a \textbf{relative} fairness protocol. In this section, we demonstrate how to construct dependency graphs between transactions and derive a final transaction ordering from the dependency graphs.

\shortversion{
\begin{figure}[t]
\begin{myprotocol}
    \STATE \textbf{\Changed{State Variables}}
    \STATE $graphs \GETS []$
    \STATE $current\_round$, $replica\_id$ \STATE \Changed{$counter \GETS 0$}
    \SPACE

    \TITLE{Client Thread}{processing transactions from clients}
    \EVENT{\Changed{On receive} a transaction $T$}
        \IF{$T$ is valid}
            \STATE \Changed{$counter \GETS counter + 1$}\label{alg:rl-counter}
            \STATE \Changed{$ois.append(counter)$,\quad $dgs.append(T.digest)$}\label{alg:rl-counter1}
        \ENDIF
    \ENDEVENT
    \SPACE

    \TITLE{DAG Layer Thread}{forming the DAG}
    \EVENT{\Changed{On propose vertex}}
        \STATE $v \GETS DAGVertex(replica\_id, current\_round)$
        \STATE Add pending local ordering $dgs$ and $ois$ into $v$
        \STATE Clear $dgs$ and $ois$
        \STATE Reliably broadcast $v$
    \ENDEVENT

    \SPACE
    \EVENT{On commit($L_r$)\label{alg:rl-commit}}
        \STATE Send $A_r$ to Ordering Layer\label{alg:rl-commit1}
    \ENDEVENT
\end{myprotocol}
\caption{FairDAG-RL: Transaction dissemination and DAG vertex proposal.}
\label{fig:rl1}
\end{figure}}

\subsection{Transaction Dissemination}\label{ss:rlvertex}

As in \FairDAG{}, clients broadcast transactions to all replicas for censorship resistance. Upon meeting the conditions in DAG protocols, each replica broadcasts a vertex. A \FairDAGRL{} vertex is a restricted form of a \FairDAG{} vertex: it includes a sequence of incrementing counter values as ordering indicators corresponding to received transactions (Figure~\ref{fig:rl1} Lines~\ref{alg:rl-counter}-\ref{alg:rl-counter1}). For simplicity, we omit the ordering indicators in Figure~\ref{fig:graph}.

The local orderings in subdag $A_r$ will be input to the fairness layer after $L_r$ is committed in \FairDAGRL{} (Lines~\ref{alg:rl-commit}-\ref{alg:rl-commit1}).

\subsection{Dependency Graph Construction}

The fairness layer of \FairDAGRL{} utilizes local orderings in committed DAG vertices to construct dependency graphs that reflect the ordering preferences of replicas. Every time a leader vertex is committed, a new dependency graph is constructed:

\fullversion{
\begin{figure}[t]
    \shortversion{\vspace{-12mm}}
\begin{myprotocol}
    \STATE \textbf{\Changed{State Variables}}
    \STATE $graphs \GETS []$
    \STATE $current\_round$, $replica\_id$ \STATE \Changed{$counter \GETS 0$}
    \SPACE

    \TITLE{Client Thread}{processing transactions from clients}
    \EVENT{\Changed{On receive} a transaction $T$}
        \IF{$T$ is valid}
            \STATE \Changed{$counter \GETS counter + 1$}\label{alg:rl-counter}
            \STATE \Changed{$ois.append(counter)$,\quad $dgs.append(T.digest)$}\label{alg:rl-counter1}
        \ENDIF
    \ENDEVENT
    \SPACE

    \TITLE{DAG Layer Thread}{forming the DAG}
    \EVENT{\Changed{On propose vertex}}
        \STATE $v \GETS DAGVertex(replica\_id, current\_round)$
        \STATE $v.dgs \GETS dgs$,\quad $v.ois \GETS ois$
        \STATE $dgs.clear()$,\quad $ois.clear()$
        \STATE Reliably broadcast $v$
    \ENDEVENT

    \SPACE
    \EVENT{On commit($L_r$)\label{alg:rl-commit}}
        \STATE Send $A_r$ to Ordering Layer\label{alg:rl-commit1}
    \ENDEVENT
\end{myprotocol}
\caption{FairDAG-RL: Transaction dissemination and DAG vertex proposal.}
\label{fig:rl1}
\end{figure}}

First, transaction digests from $A_r$ are added as graph nodes if it has not been added to any dependency graph is previously rounds. Each transaction digest $d$ is associated with a node that stores:

\begin{itemize}
    \item $type$: the type of the transaction node (See details later in \textbf{Adding nodes})
    \item $committed\_ois$: a vector of committed ordering indicators for $d$.
    \item $committed\_rounds$: the rounds in which the corresponding $commited\_oi$ is committed.
    \item $G$: the graph to which the node is added.
\end{itemize}

Second, pairwise ordering preferences are aggregated into a weight function: $weight(d_1, d_2)$ denotes the number of replicas who have a lower $committed\_oi$ for $d_1$ than $d_2$. Each dependency graph contains the following information: 

\begin{itemize}
    \item $nodes$: the set of transaction digest nodes.
    \item $weight:$ mapping of the pairs of nodes to their weights.
    \item $edges$: directed edges representing inferred ordering constraints between the transaction digest nodes.
\end{itemize}

Third, a directed edge from $d_1$ to $d_2$ is added if $weight(d_1, d_2)$ exceeds a quorum-based threshold $\frac{\n-\f}{2}$ and is not lower than $weight(d_2, d_1)$.

\begin{flushleft}
    \textbf{Thresholds}
\end{flushleft}

\fullversion{
\begin{figure}[t]
\begin{equation*}
    \begin{aligned}
        A_2: & \quad 
        \begin{aligned}
            &R_1: \{d_0, d_1, d_2, d_5, d_3\} \\
            &R_2: \{d_0, d_2, d_3, d_4, d_5\} \\
            &R_3: \{d_0, d_4, d_1, d_6\} \\
            &R_4: \{d_0, d_4, d_1, d_2\}
        \end{aligned}
        \quad\quad
        A_4: & \quad 
        \begin{aligned}
            &R_1: \{d_4, d_6\} \\
            &R_2: \{d_1, d_6\} \\
            &R_3: \{d_3, d_2, d_5, d_7\} \\
            &R_4: \{d_3, d_5, d_6, d_7\}
        \end{aligned}
    \end{aligned}
\end{equation*}
    \begin{tikzpicture}[yscale=0.8]
        \node at (-1.3, 2.2) {After processing $A_2$:};
        \node at (-2, 1.5) {$G_{2}$:};
        \node[minimum width=0.4cm,draw,circle,font=\tiny, inner sep=0pt] (a) at (-2,0.15) {$d_0$};
        \node[minimum width=0.4cm,draw,circle,font=\tiny, inner sep=0pt] (b) at (-1,1.5) {$d_1$};
        \node[minimum width=0.4cm,draw,circle,font=\tiny, inner sep=0pt] (c) at (-1,0.6) {$d_2$};
        \node[minimum width=0.4cm,draw,circle,dashed,font=\tiny, inner sep=0pt] (d) at (-1,-0.3) {$d_3$};
        \node[minimum width=0.4cm,draw,circle,font=\tiny, inner sep=0pt] (e) at (-1,-1.2) {$d_4$};
        \node[minimum width=0.4cm,draw,circle, dashed,font=\tiny, inner sep=0pt] (f) at (0,0.15) {$d_5$};
        \draw[->] (a) -- (b);
        \draw[->] (a) -- (c);
        \draw[->] (a) -- (d);
        \draw[->] (a) -- (e);
        \draw[->] (a) -- (f);
        
        \draw[->] (b) to (c);
        \draw[->, bend left] (b) to (d);
        \draw[->, bend right] (e) to (b);
        \draw[->] (c) to (d);
        \draw[->, bend right] (c) to (e);
        \draw[->] (d) to (e);

        \draw[->] (b) -- (f);
        \draw[->] (c) -- (f);
        \draw[->] (e) -- (f);


\node at (2, 2.2) {After processing $A_4$:};
        \node at (1, 1.5) {$G_{2}$:};
        \node[minimum width=0.4cm,draw,circle,font=\tiny, inner sep=0pt] (a2) at (1,0.15) {$d_0$};
        \node[minimum width=0.4cm,draw,circle,font=\tiny, inner sep=0pt] (b2) at (2,1.5) {$d_1$};
        \node[minimum width=0.4cm,draw,circle,font=\tiny, inner sep=0pt] (c2) at (2,0.6) {$d_2$};
        \node[minimum width=0.4cm,draw,circle,dashed,font=\tiny, inner sep=0pt] (d2) at (2,-0.3) {$d_3$};
        \node[minimum width=0.4cm,draw,circle,font=\tiny, inner sep=0pt] (e2) at (2,-1.2) {$d_4$};
        \node[minimum width=0.4cm,draw,circle, dashed,font=\tiny, inner sep=0pt] (f2) at (3,0.15) {$d_5$};

\draw[->] (a2) -- (b2);
\draw[->] (a2) -- (c2);
\draw[->] (a2) -- (d2);
\draw[->] (a2) -- (e2);
\draw[->] (a2) -- (f2);

\draw[->] (b2) to (c2);
\draw[->, bend left] (b2) to (d2);
\draw[->, bend right] (e2) to (b2);
\draw[->] (c2) to (d2);
\draw[->, bend right] (c2) to (e2);
\draw[->] (d2) to (e2);

\draw[->] (b2) -- (f2);
\draw[->] (c2) -- (f2);
\draw[->, blue] (d2) -- (f2);
\draw[->] (e2) -- (f2);

\node at (3.5, 1.5) {$G_{4}$:};

\node[minimum width=0.4cm,draw,circle,font=\tiny, inner sep=0pt] (y) at (3.7,0.15) {$d_6$};
\node[minimum width=0.4cm,draw,circle, dashed, font=\tiny, inner sep=0pt] (z) at (4.5,0.15) {$d_7$};

\draw[->] (y) -- (z);

\node at (-2, -2.5) {$G_2^c$:};

\node[minimum width=0.4cm, minimum height=0.4cm, draw, rectangle, font=\tiny] (b1) at (-1,-2.7) {$S_1:\{d_0\}$};
\node[minimum width=0.4cm, minimum height=0.4cm, draw, rectangle, font=\tiny] (b2) at (-1,-3.5) {$S_2:\{d_1,d_2,d_3,d_4\}$};
\node[minimum width=0.4cm, minimum height=0.4cm, draw, rectangle, font=\tiny] (b3) at (-1,-4.3) {$S_3:\{d_5\}$};
\draw[->] (b1) -- (b2);
\draw[->] (b2) -- (b3);

\node at (3,-5) {(d)};
\node at (-1,-5) {(c)};
\node at (3,-1.8) {(b)};
\node at (-1,-1.8) {(a)};

\node at (3, -3) {Readding $d_5$ into $G_{4}$:};

\node[minimum width=0.4cm,draw,circle,font=\tiny, inner sep=0pt] (x2) at (1.8,-4) {$d_5$};
\node[minimum width=0.4cm,draw,circle,font=\tiny, inner sep=0pt] (y2) at (2.9,-4) {$d_6$};
\node[minimum width=0.4cm,draw,circle, dashed, font=\tiny, inner sep=0pt] (z2) at (4,-4) {$d_7$};

\draw[->] (x2) -- (y2);
\draw[->, bend left] (x2) to (z2);
\draw[->] (y2) -- (z2);

    \end{tikzpicture}
    \caption{Constructing dependency graphs and finalizing transaction order with $\n = 4, \gamma = 1, \f=1$.}
    
    \label{fig:graph}
\end{figure}}

To construct a dependency graph that reflects the ordering preferences of the majority and mitigates manipulation by Byzantine replicas, relative fairness protocols should leverage as many correct local orderings as possible. Prior protocols such as \Themis{}~\cite{themis} and \Rashnu{}~\cite{rashnu} rely on a single leader to collect local orderings, where a Byzantine leader may intentionally exclude local orderings from up to $\f$ correct replicas. As a result, a transaction may appear in at most $\n{-}2\f$ committed local orderings if Byzantine replicas ignore the transaction and the Byzantine leader excludes $\f$ correct local orderings, increasing both the susceptibility to ordering manipulation and the minimal requirement of number of correct replicas (see Section~\ref{sec:newcomparison} for details).

In contrast, DAG-based consensus protocols ensure that all correct local orderings are eventually committed. This guarantees that each transaction can appear in at least $\n - \f$ committed local orderings. Consequently, \FairDAG{} raises the thresholds used in the construction of dependency graphs from $\n - 2\f$ and $\frac{\n - 2\f}{2}$ (used in \Themis{} and \Rashnu{}) to $\n - \f$ and $\frac{\n - \f}{2}$, respectively. We will elaborate on the details below.

\begin{figure}[t]
\begin{myprotocol}
    \TITLE{Fairness Layer Thread}{Ordering Transactions}
    \EVENT{On receive $A_r$}
        \STATE $G_r \GETS NewGraph()$,\quad $graphs.push(G_r)$\label{alg:newgraph}
        
        \STATE \COMMENT{Find nodes updated with $A_r$}\label{alg:updatecommit}
        \STATE $updated\_nodes \GETS \{\}$
        
        \FOR{$v \in A_r$}
            \STATE $i \GETS v.replica\_id$
            \FOR{$(d, oi) \in (v.dgs, v.ois)$}
                \STATE $node(d).committed\_ois[i] \GETS oi$
                \STATE $node(d).committed\_rounds[i] \GETS r$
                \STATE $updated\_nodes.insert(d)$\label{alg:updatecommit1}
            \ENDFOR
        \ENDFOR

        \STATE \COMMENT{Update node types and add new non-blank nodes to $G_r$}\label{alg:addnode}
        \FOR{$d \in updated\_nodes$}
            \IF{$node(d).type = blank$}
                \STATE Let $ap(d,r)$ be the number of ordering indicators for $d$ that are committed by round $r$
                \IF{$ap(d,r) \ge \n - \f$}
                    \STATE \Changed{$node(d).type \GETS solid$; \quad $G_r.nodes.add(node(d))$}
                \ELSIF{$ap(d,r) \ge \frac{\n - \f}{2}$}
                    \STATE \Changed{$node(d).type \GETS shaded$; \quad $G_r.nodes.add(node(d))$}\label{alg:addnode1}
                \ENDIF
            \ENDIF
        \ENDFOR

        \STATE \COMMENT{Find all candidate edges in all existing graphs $G_r$}
        \STATE $addable\_edges \GETS \{\}$
        \FOR{$v \in A_r$\label{alg:incweight}}
            \STATE $i \GETS v.replica\_id$
            \FOR{$(d, oi) \in (v.dgs, v.ois)$}
                \STATE $G' \GETS node(d).G$
                \STATE $d\_oi \GETS node(d).committed\_ois[i]$
                \FOR{$node(d_2) \in G'.nodes$}
                    \IF{$d\_oi < node(d_2).committed\_ois[i]$}
                        \STATE \Changed{increment $G'.weight[(d,d_2)]$}
                    \ELSE
                        \STATE \Changed{increment $G'.weight[(d_2,d)]$}\label{alg:incweight1}
                    \ENDIF
                    \IF{\Changed{either weight reaches threshold $\frac{\n - \f}{2}$}\label{alg:addableedge}}
                        \STATE $addable\_edges.insert(d,d_2)$\label{alg:addableedge1}
                    \ENDIF
                \ENDFOR
            \ENDFOR
        \ENDFOR

        \STATE \COMMENT{Add edge if weight reaches threshold}\label{alg:addedge}
        \FOR{$(d, d_2) \in addable\_edges$}
            \STATE $G \GETS node(d).G$
            \IF{$G.weight[(d,d_2)] \ge G.weight[(d_2,d)]$}
                \STATE $G.edges.add(e(d,d_2))$
            \ELSE
                \STATE $G.edges.add(e(d_2,d))$\label{alg:addedge1}
            \ENDIF
        \ENDFOR

        \STATE \COMMENT{Finalize Transaction Ordering in $G_r$ when $G_r$ is a tournament}
        \STATE \MName{OrderFinalization()}\label{alg:rl-finalize}
    \ENDEVENT
\end{myprotocol}
\caption{Dependency graph construction in FairDAG-RL.}
\label{fig:construct}
\end{figure}

\begin{flushleft}
    \textbf{Adding nodes}
\end{flushleft}

For each transaction digest $d$ in $A_r$, let $node(d)$ represent its dependency graph node. Its $committed\_ois$ and $committed\_rounds$ are updated using information contained in $A_r$, and the node is added to $updated\_nodes$, a set of transactions whose $committed\_ois$ are updated because of $A_r$ (Figure~\ref{fig:construct} Lines~\ref{alg:updatecommit}-\ref{alg:updatecommit1}) \footnote{All operations related to dependency graph construction only apply to transactions that are not ordered yet. We omit this for simplicity in the protocol description and the pseudocode.}.

Then, we check if any transaction digests can be added as nodes into the new dependency graph of round $r$. We define $ap(d, r)$ as the number of ordering indicators for $d$ that are committed by round $r$: 
\begin{center}
    $ap(d,r) \GETS |\{i| node(d).committed\_rounds[i] \le r\}|$
\end{center}

Each node in $updated\_nodes$ whose type is $blank$, which means that it has not been added into any dependency graph previously, is classified as:

\begin{itemize}
    \item $solid$, if $ap(d, r) \ge \n-\f$.
    \item $shaded$, if $\frac{\n-\f}{2} \le ap(d,r) < \n-\f$.
    \item $blank$, if $ap(d,r) < \frac{\n-\f}{2}$.
\end{itemize}

The we add the \emph{non-blank} nodes to $G_r$ (Lines~\ref{alg:addnode}-\ref{alg:addnode1}). Since classification is only applied to previously $blank$ nodes, it is guaranteed that each digest is inserted into at most one dependency graph.

\begin{flushleft}
    \textbf{Updating weights between nodes}
\end{flushleft}

After classifying and adding nodes to $G_r$, we update edge weights based on local orderings in $A_r$. Vertices in $A_r$ are processed in round-increasing order.

For each vertex \Changed{$v$} from replica $R_i$, we iterate through the transaction digests in $v.dgs$. For each digest $d$, we compare the value of $node(d).committed\_ois[i]$ with all other nodes in the  graph $G$ of $node(d)$. For each pair $(d, d_2)$, we increment $G.weight[(d, d_2)]$ if $node(d).committed\_ois[i] < node(d_2).committed\_ois[i]$; otherwise, we increment $G.weight[(d_2, d)]$ (Lines~\ref{alg:incweight}-\ref{alg:incweight1}).

\fullversion{
\begin{figure}[t]
    \begin{myprotocol}
        \TITLE{Fairness Layer Thread}{Ordering transactions}
        \FUNCTION{OrderFinalization}{}
        \WHILE{$G_r \GETS graphs.Front()$}
            \IF{$G_r$ is a tournament}
            \STATE $graphs.Pop()$
            \STATE $G_r^c \GETS Trajan\_SCC(G_r)$
            \STATE $[S_1, S_2, ..., S_s] \GETS Topologically\_Sorted(G_r^c)$
            \SPACE
            \STATE $last \GETS \max \{ j \mid \exists node \in S_j, node.type = solid \}$
            \FOR{$j = 1, 2, ..., last$}
            \STATE $p_j \GETS Hamilton\_Path(S_j)$
            \STATE Append $p_j$ to final ordering
            \FOR{$node(d) \in p_j$}
            \STATE $ordered\_nodes.add(node(d))$
            \ENDFOR
            \ENDFOR
            \SPACE
            \STATE $G_{r'} \GETS graphs.Front()$
            \FOR{$j = last+1, last+2, ..., s$}
            \FOR{$node(d) \in S_j$}
            \STATE $\begin{aligned}ap(d, r') \GETS |\{i\mid node(d).committed\_rounds[i] \le r'\}|\end{aligned}$\IF{$ap(d, r') \ge \n-\f$}
                \STATE $node(d).type = solid$
            \ELSIF{$ap(d, r') \ge \frac{\n-\f}{2}$}
                \STATE $node(d).type = shaded$
            \ENDIF
            \STATE $G_{r'}.nodes.add(node(d))$
            \FOR{$node(d_2) \in G_{r'}.nodes$}
            \STATE Calculate $G_{r'}.weights[(d, d_2)]$
            \STATE Calculate $G_{r'}.weights[(d_2, d)]$
            \IF{$G_{r'}.weights[(d, d_2)] \ge \frac{\n-\f}{2}$\\ $\lor G_{r'}.weights[(d_2, d)] \ge \frac{\n-\f}{2}$}
            \IF{$G_{r'}.weights[(d, d_2)]\ge G_{r'}.weights[(d_2, d)]$}
            \STATE $G_{r'}.edges.add(e(d,d_2))$
            \ELSE 
            \STATE $G_{r'}.edges.add(e(d_2,d))$
            \ENDIF
            \ENDIF
            
            \ENDFOR
            \ENDFOR
            \ENDFOR
            \ELSE
            \STATE \textbf{break}
            \ENDIF
        \ENDWHILE
        \ENDFUNCTION
    \end{myprotocol}
    \caption{Finalizing Transaction Ordering in FairDAG-RL.}
    \label{fig:finalize}
\end{figure}
}

During this process, we maintain a set $addable\_edges$ to track the edges that can be added to the dependency graph. If $G.weight[(d, d_2)]$ or $G.weight[(d_2, d)]$ reaches the threshold $\frac{\n - \f}{2}$, the corresponding pair is added to $addable\_edges$ (Lines~\ref{alg:addableedge}-\ref{alg:addableedge1}).

\shortversion{\newpage}
\begin{flushleft}
    \textbf{Adding edges}
\end{flushleft}

For each pair $(d, d_2) \in addable\_edges$, if there is no edge between $node(d)$ and $node(d_2)$, an edge is added based on the majority preference (Lines~\ref{alg:addedge}-\ref{alg:addedge1}):

\begin{itemize}
\item If $G.weight[(d, d_2)] \ge G.weight[(d_2, d)]$, add edge $e(d, d_2)$ from $node(d)$ to $node(d_2)$ ;
\item Otherwise, add edge $e(d_2, d)$ from $node(d_2)$ to $node(d)$.
\end{itemize}

\shortversion{
\begin{figure}[t]
\begin{equation*}
    \begin{aligned}
        A_2: & \quad 
        \begin{aligned}
            &R_1: \{d_0, d_1, d_2, d_5, d_3\} \\
            &R_2: \{d_0, d_2, d_3, d_4, d_5\} \\
            &R_3: \{d_0, d_4, d_1, d_6\} \\
            &R_4: \{d_0, d_4, d_1, d_2\}
        \end{aligned}
        \quad\quad
        A_4: & \quad 
        \begin{aligned}
            &R_1: \{d_4, d_6\} \\
            &R_2: \{d_1, d_6\} \\
            &R_3: \{d_3, d_2, d_5, d_7\} \\
            &R_4: \{d_3, d_5, d_6, d_7\}
        \end{aligned}
    \end{aligned}
\end{equation*}
    \begin{tikzpicture}[yscale=0.8]
        \node at (-1.3, 2.2) {After processing $A_2$:};
        \node at (-2, 1.5) {$G_{2}$:};
        \node[minimum width=0.4cm,draw,circle,font=\tiny, inner sep=0pt] (a) at (-2,0.15) {$d_0$};
        \node[minimum width=0.4cm,draw,circle,font=\tiny, inner sep=0pt] (b) at (-1,1.5) {$d_1$};
        \node[minimum width=0.4cm,draw,circle,font=\tiny, inner sep=0pt] (c) at (-1,0.6) {$d_2$};
        \node[minimum width=0.4cm,draw,circle,dashed,font=\tiny, inner sep=0pt] (d) at (-1,-0.3) {$d_3$};
        \node[minimum width=0.4cm,draw,circle,font=\tiny, inner sep=0pt] (e) at (-1,-1.2) {$d_4$};
        \node[minimum width=0.4cm,draw,circle, dashed,font=\tiny, inner sep=0pt] (f) at (0,0.15) {$d_5$};
        \draw[->] (a) -- (b);
        \draw[->] (a) -- (c);
        \draw[->] (a) -- (d);
        \draw[->] (a) -- (e);
        \draw[->] (a) -- (f);
        
        \draw[->] (b) to (c);
        \draw[->, bend left] (b) to (d);
        \draw[->, bend right] (e) to (b);
        \draw[->] (c) to (d);
        \draw[->, bend right] (c) to (e);
        \draw[->] (d) to (e);

        \draw[->] (b) -- (f);
        \draw[->] (c) -- (f);
        \draw[->] (e) -- (f);


\node at (2, 2.2) {After processing $A_4$:};
        \node at (1, 1.5) {$G_{2}$:};
        \node[minimum width=0.4cm,draw,circle,font=\tiny, inner sep=0pt] (a2) at (1,0.15) {$d_0$};
        \node[minimum width=0.4cm,draw,circle,font=\tiny, inner sep=0pt] (b2) at (2,1.5) {$d_1$};
        \node[minimum width=0.4cm,draw,circle,font=\tiny, inner sep=0pt] (c2) at (2,0.6) {$d_2$};
        \node[minimum width=0.4cm,draw,circle,dashed,font=\tiny, inner sep=0pt] (d2) at (2,-0.3) {$d_3$};
        \node[minimum width=0.4cm,draw,circle,font=\tiny, inner sep=0pt] (e2) at (2,-1.2) {$d_4$};
        \node[minimum width=0.4cm,draw,circle, dashed,font=\tiny, inner sep=0pt] (f2) at (3,0.15) {$d_5$};

\draw[->] (a2) -- (b2);
\draw[->] (a2) -- (c2);
\draw[->] (a2) -- (d2);
\draw[->] (a2) -- (e2);
\draw[->] (a2) -- (f2);

\draw[->] (b2) to (c2);
\draw[->, bend left] (b2) to (d2);
\draw[->, bend right] (e2) to (b2);
\draw[->] (c2) to (d2);
\draw[->, bend right] (c2) to (e2);
\draw[->] (d2) to (e2);

\draw[->] (b2) -- (f2);
\draw[->] (c2) -- (f2);
\draw[->, blue] (d2) -- (f2);
\draw[->] (e2) -- (f2);

\node at (3.5, 1.5) {$G_{4}$:};

\node[minimum width=0.4cm,draw,circle,font=\tiny, inner sep=0pt] (y) at (3.7,0.15) {$d_6$};
\node[minimum width=0.4cm,draw,circle, dashed, font=\tiny, inner sep=0pt] (z) at (4.5,0.15) {$d_7$};

\draw[->] (y) -- (z);

\node at (3,-1.8) {(b)};
\node at (-1,-1.8) {(a)};
    \end{tikzpicture}
    \caption{Constructing dependency graphs with $\n = 4, \gamma = 1, \f=1$.}
    \label{fig:graph}
\end{figure}}

\begin{example}
Figure~\ref{fig:graph} illustrates how \FairDAGRL{} constructs dependency graphs. 
In Figure~\ref{fig:graph}(a), after processing $A_2$, graph $G_2$ contains six nodes. Nodes $d_0$, $d_1$, $d_2$, and $d_4$ are classified as \emph{solid} (solid circles), while $d_3$ and $d_5$ are \emph{shaded} (dashed circles). \Changed{All node pairs form edges except $(d_3, d_5)$,} as neither $G_2.weight[(d_3, d_5)]$ nor $G_2.weight[(d_5, d_3)]$ reaches the threshold $\frac{\n - \f}{2} = \Changed{\frac{3}{2}}$. In Figure~\ref{fig:graph}(b), after processing $A_4$, $G_4$ is constructed with two additional nodes, and an edge between $node(d_3)$ and $node(d_5)$ is added once $G_2.weight[(d_3, d_5)]$ reaches the threshold.
\end{example}

\shortversion{
\subsection{Ordering Finalization}

\FairDAGRL{} finalizes the ordering of transactions within a dependency graph after it becomes a \emph{tournament}. A \emph{tournament} is a dependency graph such that there is an edge between each pair of transaction nodes in the dependency graph. When finalizing the ordering, \FairDAGRL{} condenses the dependency graph into multiple \emph{Strongly Connected Components (SCCs)} and generates a topological sorting of them. Transactions that are sorted behind the last \emph{SCC} containing at least one \emph{solid} transaction will be readded into later dependency graphs, as they might be dependent on some transactions in later dependency graphs (see more detailed explanation in Section 8.4 of our extended report~\cite{extended-report}.)

As we use the same algorithm to finalize transaction ordering within \emph{tournaments} as \Themis{}, we put the detailed algorithm in Section 6.3 our extended report~\cite{extended-report}.
}

\fullversion{
\subsection{Ordering Finalization}

A \emph{tournament} graph is a graph such that there is an edge between each pair of nodes. After adding edges, we check if the graphs are \emph{tournaments} in a round-increasing order. If a graph is a \emph{tournament}, we finalize the ordering of transactions within it as follows and then check the next graph until encountering a \emph{non-tournament} graph.

\begin{flushleft}
    \textbf{Condensing dependency graph}
\end{flushleft}

In graph theory, a \emph{strongly connected component (SCC)} is a maximal subset of nodes such that for every pair of nodes $(node_1, node_2)$ in the subset, there exists a directed path from $node_1$ to $node_2$ and a directed path from $node_2$ to $node_1$.

For a \emph{tournament} dependency graph $G_r$, we condense it into $G_r^c$ by finding all \emph{SCCs} with \emph{Tarjan's SCC algorithm}. Figure~\ref{fig:graph} (c) shows $G_2^c$ the condensed graph of $G_2$.

\begin{flushleft}
    \textbf{Ordering finalization}
\end{flushleft}

It is guaranteed that after condensation, there will be one and only one topological sorting of the \emph{SCCs} in $G_r^c$, which we denote by $S_1, S_2, ..., S_{s}$, where $S_j$ represents the $j$-th \emph{SCC} in the topological sorting.

We denote by the $S_{last}$ the last \emph{SCC} that contains a \emph{solid} node. Then, for each $S_j$ such that $j \le last$, we find a \emph{Hamilton path} $p_j$ of nodes in $S_j$ and append $p_j$ to the final transaction ordering. In Figure~\ref{fig:graph} (c), $S_2$ is $S_last$ of $G_2^c$.

\begin{flushleft}
    \textbf{Readding shaded nodes}
\end{flushleft}

For the \emph{SCCs} behind $S_j$, the nodes are all \emph{shaded} nodes. We add these nodes into the next graph $G_{r'}$. For each $node(d)$, we run the following steps: 

\begin{enumerate}
    \item Classify $node(d)$ based on $ap(d, r')$;
    \item Calculate the weights between $node(d)$ and all existing nodes in $G_{r'}$.
    \item Add edges between $node(d)$ and other nodes if the weight reaches threshold $\frac{\n-\f}{2}$.
\end{enumerate}

In Figure~\ref{fig:graph} (d), $node(d_5)$ is readded into $G_4$, being classified as a \emph{solid} node. And two edges pointing from $node(d_5)$ are added.

\subsection{Ordering Dependency}
    After illustrating how to construct dependency graphs and finalize transaction ordering in \FairDAGRL{}, we now discuss \emph{ordering dependencies} in \FairDAGRL{}. For any transaction $T_1$ and $T_2$ with digests $d_1$ and $d_2$, either $T_1$ is dependent on $T_2$ or $T_2$ is dependent on $T_1$. As the transactions are ordered based on the edges in the dependency graphs, intuitively, if an edge $e(d_1, d_2)$ exists, then $T_2$ is dependent on $T_1$. 
    
    However, for two transactions that are in different replicas, deciding \emph{ordering dependency} is more complicated. We denote by $weight[(d_2, d_1)]_{max}$ the maximal number of the local orderings in which $T_2$ is ordered before $T_1$. Even though $node(d_2)$ is in an earlier dependency graph, it is possible that $weight[(d_2, d_1)]_{max} < \frac{\n-\f}{2}$, which implies that edge $e(d_2, d_1)$ cannot exist even if the two nodes are in the same dependency graph. Then, if $weight[(d_2, d_1)]_{max} < \frac{\n-\f}{2}$, $T_2$ is dependent on $T_1$. 
    
    To endorse the transactions that are added as nodes earlier, if $T_1$ is in a earlier dependency graph and $weight[(d_1, d_2)]_{max} \ge \frac{\n-\f}{2}$, which implies that edge $e(d_1,d_2)$ could exist, we also say that $T_1$ is also dependent on $T_2$.

    \begin{definition}
        We say that $T_2$ with digest $d_2$ is dependent on $T_1$  with digest $d_1$ if:
        \begin{itemize}
            \item $weight[(d_2, d_1)]_{max} < \frac{\n-\f}{2}$; \textbf{or}
            \item edge $e(d_1, d_2)$ exists; \textbf{or}
            \item $weight[(d_1, d_2)]_{max} >= \frac{\n-\f}{2}$ and $node(d_1)$ is in a earlier dependency graph.
        \end{itemize}
    \end{definition}


}
\section{Comparing Fairness Protocols}~\label{sec:newcomparison}

In this section, we demonstrate how \FairDAG{} and \FairDAGRL{} outperform \Pompe{}~\cite{pompe} and \Themis{}~\cite{themis} in limiting the adversary's manipulation of transaction ordering. We achieve this by comparing how these protocols perform under adversarial conditions, such as those caused by Byzantine replicas or an asynchronous network.

\subsection{Pompe and Themis}\label{ss:intropompethemis}

\Pompe{} and \Themis{} run fairness protocols atop leader-based consensus protocols such as PBFT~\cite{pbftj} and HotStuff~\cite{hotstuff}, where a single leader is responsible for collecting local orderings from a quorum of $\n{-}\f$ replicas.

Beyond the difference in underlying consensus protocols, \Pompe{} assigns non-overlapping intervals to consensus rounds, allowing only transactions whose $AOI$ fall within the corresponding interval to be executed, increasing the overhead of recovery when the leader fails. Additionally, in \Pompe{}, each client transaction is sent to a single replica instead of being broadcast, making the protocol more susceptible to transaction censorship.

\subsection{Adversarial Manipulation}\label{ss:adversarial}

We now analyze how an adversary can manipulate transaction ordering in \Pompe{} and \Themis{} by selectively collecting local orderings. Additionally, we demonstrate how \FairDAG{} and \FairDAGRL{} mitigate these vulnerabilities, ensuring more resilient and fair transaction orderings.

\begin{flushleft}
    \textbf{Pompe vs FairDAG-AB}
\end{flushleft}

In \Pompe{}, a Byzantine replica who is responsible for collecting ordering indicators can selectively choose $\n{-}\f$ local ordering indicators to calculate an assigned ordering indicator. This selective collection allows the adversary to manipulate the ordering of transactions whose ranges of correct ordering indicators overlap.

Furthermore, if at most $\f$ correct replicas have received the assigned ordering indicator for a transaction $T$ before a Byzantine leader replica collects assigned ordering indicators for a new block, the leader can exclude $T$ from the new block. Since the leader is only required to collect from $\n{-}\f$ replicas, it can selectively collect from $\n{-}\f$ replicas that have not received $T$, thus delaying the position of $T$ in the final ordering.

In \FairDAG{}, the client broadcasts its transaction to all replicas, and each replica independently broadcasts its ordering indicators.
The \emph{Validity} property and the round-robin or random leader rotation of the underlying DAG protocols guarantees that reduce the chance of selective collection of ordering indicators by Byzantine replicas. Thus, \FairDAG{} is more resilient against ordering manipulation than \Pompe{}.

Additionally, 
\Pompe{} effectively mitigates the transaction censorship issue mentioned in Section~\ref{ss:intropompethemis}, because transactions are broadcast to all replicas, and each replica independently generates and broadcasts its local ordering indicators.

\begin{flushleft}
    \textbf{Themis vs FairDAG-RL}
\end{flushleft}

In \Themis{}, if a Byzantine leader ignore $\f$ correct local orderings containing transaction $T$ and the Byzantine replicas exclude $T$ from their local orderings, there would be at most $\n{-}2\f$ local orderings containing $T$. Thus, the threshold of deciding edge direction is $\frac{\n{-}2\f}{2}$ in \Themis{}. To guarantee the \emph{$\gamma$-Batch-Order-Fairness}, it is required that the votes of the opposite direction cannot reach the threshold, containing $\f$ votes from Byzantine replicas and $(1{-}\gamma)(\n{-}\f)$ votes from correct replicas. That is,
$\f+ (1{-}\gamma)(\n{-}\f) < \frac{\n{-}2\f}{2}$, i.e., $\n>\frac{\f(2\gamma+2)}{2\gamma-1}$. When $\gamma=1$, \Themis{} requires $\n>4\f$.

In \FairDAGRL{}, due to the Validity of DAG-based consensus protocols, all local orderings from correct replicas will eventually be committed by all replicas. Then, for each transaction, there are at least $\n{-}\f$ local orderings containing it. Thus, the threshold for deciding edge direction is $\frac{\n{-}\f}{2}$ in \Themis{}. Then, to guarantee the \emph{$\gamma$-Batch-Order-Fairness}, it is required that $\f+ (1{-}\gamma)(\n{-}\f) < \frac{\n{-}\f}{2}$, i.e., $\n>\frac{\f(2\gamma+1)}{2\gamma-1}$.  When $\gamma=1$, \FairDAGRL{} requires $\n>3\f$.

\subsection{Leader Crash and Asynchronous Network}

In \Pompe{} and \Themis{}, if the leader crashes, a recovery subprocess must be initiated to replace the leader with a new one, introducing an additional delay of $O(\Delta)$ before the protocol can resume normal operation. Moreover,
in \Pompe{}, each round corresponds to a distinct, non-overlapping time slot. If the designated leader crashes, transactions with assigned ordering indicators that fall within that time slot cannot be ordered or executed, resulting in a potential transaction loss or indefinite delays. In \Themis{}, if the leader crashes, replicas have to resend their local orderings to the new leader, resulting in additional communication overhead. 

If the network operates under asynchronous conditions where messages can experience indefinite delays, additional overhead will be introduced, similar to the overhead incurred during leader crashes.

\FairDAG{} and \FairDAGRL{} address the aforementioned issues through the multi-proposer design and the Reliable Broadcast inherent to DAG-based consensus protocols, which ensures that each correct local ordering eventually deliviers and commits even in asynchronous settings.
\fullversion{
\section{Correctness Proof}~\label{sec:proof}

In this section, we will prove the safety, liveness and fairness properties of \FairDAG{} and \FairDAGRL{}. Derived from \cite{dagrider, narwhal, bullshark}, the DAG-Layer has the following properties:

\begin{itemize}
    \item \textbf{Agreement}: if a correct replica commits $A_r$ of leader vertex $L_r$, then every other correct replica eventually commits $A_r$.
    \item \textbf{Total Ascending Order}: if a correct replica commits $A_r$ before $A_{r'}$, then $r<r'$ and no correct replica commits $A_{r'}$ before $A_r$.
    \item \textbf{Validity}: if a correct replica broadcasts a DAG vertex $v$, then eventually each correct replica will commit a leader vertex $L_r$ such that $v \in A_r$.
\end{itemize}

Combining the \textbf{Agreement} and \textbf{Total Ascending Order}, we have the following lemma:

\begin{lemma}\label{lemma:samecausalhistory}
    If a correct replica $R$ commits a series of leader vertices $L\_list^{R} = L_{r_j}^{R}, L_{r_{j+1}}^{R}, ...$ and every other correct replica $R'$ will eventually commit $L\_list^{R'} = L_{r_j}^{R'}, L_{r_{i+1}}^{R'}, ...$, such that $\forall j$, $r_j < r_{j+1}$; $L_{r_j}^{R} = L_{r_j}^{R'}$; and $C_{r_j}^{R} = C_{r_j}^{R'}$.
\end{lemma}

\subsection{Safety}\label{ssec:safety}

\begin{lemma}\label{lemma:sameaoi}
    In \FairDAG{}, if a correct replica $R$ assigns transaction $T$ with digest $d$ \emph{assigned ordering indicators} $d.AOI^R$, then every other correct replica $R'$ will eventually assign $T$ with $d.AOI^{R'}$ such that $d.AOI^R = d.AOI^{R'}$.
\end{lemma}

\begin{proof}
    According to \FairDAG{} algorithm, we know that a correct replica deterministically calculates $d.AOI$ based on $C_r$, the causal history of the lowest-round leader vertex $L_r$ such that after the replica commits $L_r$, there are at least $\n-\f$ $committed\_ois$ of $d$ in $C_r$. 

    Then, we prove the lemma by contradiction. We assume that $R$ calculates $d.AOI^R$ based on the causal history $C_r^R$ of leader vertex $L_r^R$. If $R'$ never assigns $T$ with an \emph{assigned ordering indicator}, then $R'$ never commits $C_r^R$. If $R'$ assigns $T$ with an \emph{assigned ordering indicator} different from $T'$, then $R'$ must have committed some different $C_r'^{R'}$. Both cases contradict \ref{lemma:samecausalhistory}.
\end{proof}

\begin{lemma}\label{lemma:lpaoi}
    In \FairDAG{}, after $A_r$ is committed and processed, for a transaction $T$ with digest $d$ that has no \emph{assigned ordering indicator} but a \emph{lowest possible assigned ordering indicator} $d.LPAOI_r$, if $T$ eventually gets an \emph{assigned ordering indicator} $d.AOI$, then $d.AOI \ge d.LPAOI_r$.
\end{lemma}

\begin{proof}
    We denote by $\hts{}_r$ and $d.seen\_ois_r$, respectively, the $\hts{}$ and $d.seen\_ois$ after $A_r$ is committed.
    According to the \FairDAG{} algorithm, $d.LPOAI^r$ is the $(\f+1)$-th lowest value of $lp\_ois_r$ where $lp\_ois_r[i] \GETS $ \\ $min(d.seen\_ois_r[i], \hts{}_r[i])$, $1\le i\le \n$. As the DAG grows, each new $d.committed[i]$ will be not smaller than \\$min(d.seen\_ois_r[i], \hts{}_r[i])$. Thus, as $d.AOI \GETS $\\$sorted(d.committed\_ois[i])[\f+1]$, where $\forall i, d.committed\_ois[i] \ge lp\_ois_r[i]$, it holds true that $d.AOI \ge d.LPAOI_r$.
\end{proof}

\begin{lemma}\label{lemma:followingaoi}
    In \FairDAG{}, for any two transactions $T_1$ and $T_2$ with digests $d_1$ and $d_2$, if $d_1.AOI < d_2.AOI$, then $d_1$ will be ordered before $d_2$ in the final ordering.
\end{lemma}

\begin{proof}
    We denote by $or_j$ the round such that transaction digest $d_j$ is ordered, getting an \emph{assigned ordering indicator} lower than $LPAOI_{min}$, i.e., $d_j.AOI < LPAOI_{min}$. 
    
    Now we prove by the lemma by contradiction. Assuming that $d_2$ is ordered before $d_1$, then obviously, $or_1\ge or_2$.
    \begin{itemize}
        \item If $or_1 = or_2$, then according to \FairDAG{} algorithm, $d_1$ and $d_2$ are ordered based on assigned ordering indicator. Thus, with a lower $AOI$, $d_1$ would be ordered before $d_2$, contradicting with our assumption.
        \item If $or_1 > or_2$, then by the definition of $LPAOI_{min}$ we know that $d_2.AOI < LPAOI_{min} \le d_1.LPAOI_{or_2}$. Also, from\\ Lemma~\ref{lemma:lpaoi} we know that eventually $d_1$ will get $d_1.AOI \ge d_1.LPAOI_{or_2}$. Thus, we have $d_1.AOI > d_2.AOI$, which contradicts the fact that $d_1.AOI < d_2.AOI$. 
    \end{itemize}
    In summary, it holds true that if $d_1.AOI < d_2.AOI$, then $d_1$ will be ordered before $d_2$ in the final ordering.
\end{proof}

\begin{theorem}
    (\FairDAG{} SAFETY) In \FairDAG{}, if a correct replica orders transaction $T$ at position $p$ in the final ordering, then every correct replica will eventually order $T$ at position $p$.
\end{theorem}

\begin{proof}
    From Lemma~\ref{lemma:sameaoi} we know that every correct replica will eventually assign the same \emph{assigned ordering indicator} to $T$. From Lemma~\ref{lemma:followingaoi} we know that all transactions with \emph{assigned ordering indicators} are ordered in an ascending order of $AOI$. Combining the two claims above, we conclude that every correct replica will eventually order $T$ at the same position in the final ordering.
\end{proof}

\begin{theorem}
    (\FairDAGRL{} SAFETY) In \FairDAGRL{}, if a correct replica orders transaction $T$ at position $p$ in the final ordering, then every correct replica will eventually order $T$ at position $p$.
\end{theorem}

\begin{proof}
    After committing a leader vertex, each correct replica uses a deterministic method to construct dependency graphs and finalize transaction order. Combining this with Lemma~\ref{lemma:samecausalhistory}, we know that safety holds for \FairDAGRL{}.
\end{proof}

\subsection{Liveness}\label{ssec:liveness}

We claim the following assumption holds, which is necessary for the liveness property.

\begin{assumption}\label{asmpt:allcorrect}
    If a correct replica $R$ receives transaction $T$, then every correct replica will eventually receive transaction $T$.
\end{assumption}

\begin{lemma}\label{lemma:goodaoi}
    In \FairDAG{}, if a transaction $T$ with transaction $d$ is received by correct replicas, then $T$ will eventually get an \emph{assigned ordering indicator}.
\end{lemma}

From Assumption~\ref{asmpt:allcorrect} we know that all correct replicas will receive $T$ will propose a DAG vertex containing $d$ and a corresponding ordering indicator. From the \textbf{Validity} of DAG layer, we know that all these DAG vertices will be committed. Thus, $d$ will get at least $\n-\f$ $committed\_ois$ from correct replicas and then get an \emph{assigned ordering indicator}.

\begin{lemma}\label{lemma:nohigherlpaoi}
    For transactions $T_1$ and $T_2$ with digests $d_1$ and $d_2$, if only $T_1$ has an \emph{assigned ordering indicator} and $T_2$ has an $LPAOI$ lower than $d_1.AOI$, i.e., $d_2.AOI = \infty \land d_2.LPAOI < d_1$, then eventually, 
    \begin{itemize}
        \item either $T_2$ gets an \emph{assigned ordering indicator};
        \item or $d_2.LPAOI$ becomes larger than $d_1.AOI$.
    \end{itemize}
\end{lemma}

\begin{proof}
    If $T_2$ is received by correct replicas, from Lemma~\ref{lemma:goodaoi}, we know that eventually $T_2$ gets an \emph{assigned ordering indicator}.

    Next, we discuss the case that $T_2$ is received by only faulty replicas. Since $d_2.LPAOI$ is the $(\f+1)$-th lowest value of $lp\_ois$, $\n-\f$ values of which are the $\hts{}$ values from correct replicas, thus, $d_2.LPAOI$ is not greater than  $\hts{}$ value from at least one correct replica. Due to the \textbf{Valifity} property of the DAG layer, DAG keeps growing and the $\hts{}$ value from each correct replica will eventually be higher than $d_1.AOI$.
\end{proof}

\begin{theorem}
    (\FairDAG{} Liveness) If a transaction $T$ with digest $d$ is received by correct replicas, then $T$ will eventually be ordered.
\end{theorem}

\begin{proof}
    From Lemma~\ref{lemma:goodaoi} we know that $T$ will eventually get an $d.AOI$. From Lemma~\ref{lemma:nohigherlpaoi} we know that eventually there will be no other transaction that has an $LPAOI$  lower than $d.AOI$. Thus, eventually it will be satisfied that $d_1.AOI < LPAOI_{min}$, and then $T$ will be ordered.
\end{proof}

\begin{lemma}\label{lemma:tournament}
    In \FairDAGRL{}, each dependency graph $G$ will eventually become a tournament.
\end{lemma}

\begin{proof}
    Since only \emph{solid} and \emph{shaded} nodes can be added into a dependency graph, for each $node(d)$ in $G$, there are at least $\frac{\n-\f}{2} > \f+1$ local orderings that contain $d$. Thus, $d$ will be received by all correct replicas. And then eventually, due to the \textbf{Validity} property of the DAG layer, there will be a round $r$ such that $ap(d,r) \ge \n-\f$.

    Thus, for each pair of nodes $node(d_1)$ and $node_2$ in $G$, there will be a round $r$ such that $ap(d_1,r) \ge \n-\f$ and $ap(d_2,r) \ge \n-\f$. Thus, at least one of $G.weights[(d_1, d_2)]$ and $G.weights[(d_2, d_1)]$ will reach the threshold $\frac{\n-\f}{2}$. Then, eventually, there will be an edge between each pair of nodes in $G$, i.e., $G$ will be a tournament.
\end{proof}

\begin{theorem}
    (\FairDAGRL{} Liveness) If a transaction $T$ with digest $d$ is received by correct replicas, then $T$ will eventually be ordered.
\end{theorem}

\begin{proof}
    Due to the \textbf{Validity} property of the DAG layer, there will eventually be a round $r$ such that $ap(d,r) \ge \n-\f$, and then $node(d)$ will be added into a dependency graph $G$.

    From Lemma~\ref{lemma:tournament} we know that $G$ will eventually be a tournament. If $node(d)$ is \emph{solid}, the $T$ will be ordered. If $node(d)$ is \emph{shaded}, the $T$ might be ordered. Even if \emph{shaded} $node(d)$ has to be added into the next dependency graph, eventually $node(d)$ will be added as a \emph{solid} node and $T$ will be ordered.
\end{proof}
    
\subsection{FairDAB-AB: Ordering Linearizability}\label{ssec:fairness}

To prove ordering linearizability, we introduce the following denotations:

\begin{itemize}
    \item $d.\_ois^C$: the ordering indicators of $d$ from correct replicas.
    \item $d.low\_oi^C$: the lowest value in $d.\_ois^C$.
    \item $d.high\_oi^C$: the highest value in $d.\_ois^C$.
\end{itemize}

\begin{lemma}\label{lemma:aoirange}
    For each transaction $T$ with digest $d$, $d.low\_oi^C \le d.AOI \le d.high\_oi^C$.
\end{lemma}

\begin{proof}
    As the $d.AOI$ is the $(\f+1)$-th lowest value of a subset of $d.committed\_ois$ with at least $\n-\f \ge 2\f+1$ values. Among the $\f+1$ lowest values of the subset, at least one is in $d.\_ois^C$. Thus, $d.low\_oi^C \le d.AOI$. Similarly, among the $\f+1$ highest values of the subset, at least one is from a correct replica. Thus, $d.AOI \le d.high\_oi^C$. 
\end{proof}

\begin{theorem}
    (Ordering Linearizability) For two transactions $T_1$ and $T_2$ with digests $d_1$ and $d_2$, if $d_1.high\_oi^C < d_2.low\_oi^C$, then $T_1$ will be ordered before $T_2$ in the final ordering.
\end{theorem}

\begin{proof}
    From Lemma~\ref{lemma:aoirange} we know that $d_1.AOI \le d_1.high\_oi^C$ and  $d_2.low\_oi^C \le d_2.AOI$. Thus, $d_1.AOI < d_2.AOI$. From Lemma~\ref{lemma:followingaoi} we know that transactions are ordered based on $AOI$ values. Thus, $T_1$ will be ordered before $T_2$ in the final ordering.
\end{proof}

\subsection{FairDAG-RL: $\gamma$-Batch-Order-Fairness}\label{ssec:batchorderfairness}

\begin{lemma}\label{lemma:edgedirection}
    For any two transactions $T_1$ and $T_2$ with digests $d_1$ and $d_2$, if $\gamma(\n-\f)$ correct replicas receive $T_1$ before $T_2$, then $\frac{\n-\f}{2} > weight[(d_2, d_1)]_{max}$.
\end{lemma}

\begin{proof}
    As $\gamma(\n-\f)$ correct replicas receive $T_1$ before $T_2$, then $weight[(d_2, d_1)]_{max} \ge \f+(1-\gamma)(\n-\f)$. As $\n > \frac{(2\gamma+1)\f}{(2\gamma-1)}, \frac{1}{2} <\gamma \le 1$, we have:

    \begin{flushleft}
        $\n > \frac{(2\gamma+1)\f}{(2\gamma-1)} \iff (2\gamma-1) (\n-\f) > 2\f \iff (2\gamma-2)(\n-\f) + \n-\f > 2\f \iff \n-\f > 2\f+(2-2\gamma)(\n-\f) \iff \frac{\n-\f}{2} > \f+(1-\gamma)(\n-\f) \iff \frac{\n-\f}{2} > weight[(d_2, d_1)]_{max}$
    \end{flushleft}
\end{proof}

Combing Lemma~\ref{lemma:edgedirection} with the definition of \emph{ordering dependency}, we have: 

\begin{lemma}\label{lemma:dependent}
    For any two transactions $T_1$ and $T_2$ with digests $d_1$ and $d_2$, if $\gamma(\n-\f)$ correct replicas receive $T_1$ before $T_2$, then $T_2$ is dependent on $T_1$, i.e., $T_1 \rightarrow T_2$.
\end{lemma}

\begin{lemma}\label{lemma:earlierscc}
    For any two transactions $T_1$ and $T_2$ with digests $d_1$ and $d_2$ in a tournament dependency graph $G$, if $\gamma(\n-\f)$ correct replicas receive $T_1$ before $T_2$, then after condensing $G$ and topologically sorting the \emph{SCCs}, $node(d_1)$ is in the same or an earlier \emph{SCC} than $node(d_2)$.
\end{lemma}

\begin{proof}
    According to graph theory, we know that a tournament dependency graph $G$ can be condensed into a graph with multiple \emph{SCCs}, and after topologically sorting the \emph{SCCs}, there is a unique list $S_1, S_2, ..., S_s$. It holds that for any two \emph{SCCs} $S_{a}$ and $S_{b}$, if $a < b$, then $\forall node(d_a) \in S_{a}, \forall node(d_b) \in S_{b}$, edge $e(d_a, d_b)$ exists in $G$. 

    We prove the lemma by contradiction. Assuming that $node(d_2)$ is in an earlier \emph{SCC} than $node(d_1)$, then $e(d_2, d_1)$ exists in $G$. However, it contradicts Lemma~\ref{lemma:edgedirection}, from which we know $\frac{\n-\f}{2} > weight[(d_2, d_1)]_{max}$, i.e., $e(d_2, d_1)$ cannot exist.
\end{proof}

\begin{lemma}\label{lemma:earliergraph}
    $\forall node(d_1) \in G_{r_1}, \forall node(d_1) \in G_{r_2}, r_1 < r_2$, if $node(d_1)$ is \emph{solid}, then transaction $T_2$ is dependent on $T_1$.
\end{lemma}

\begin{proof}
    If $node(d_2)$ was not added as a node into $G_{r_1}$, then $ap(d_2, r_1) < \frac{\n-\f}{2}$ when $ap(d_1, r_1) \ge \n-\f$. Thus, there are more than $\frac{\n-\f}{2}$ committed local orderings in which $T_1$ is before $T_2$. Then $weight[(d_1, d_2)]_{max} \ge \frac{\n-\f}{2}$, then, combining with the fact that $node(d_1)$ is in an earlier dependency graph, $T_2$ is dependent on $T_1$.

    If $node(d_2)$ was added as a node into $G_{r_1}$ but readded into $G_{r_1}$, then it implies that $node(d_2)$ was in an \emph{SCC} later than the last \emph{SCC} that contains a \emph{solid} node. Thus, solid $node(d_1)$ has an edge to $node(d_2)$, i.e., $e(d_1, d_2)$ exists, and then $T_2$ is dependent on $T_1$.
\end{proof}

\begin{theorem}
    ($\gamma$-Batch-Order-Fairness) For any two transactions $T_1$ and $T_2$ with digests $d_1$ and $d_2$, if $\gamma(\n-\f)$ correct replicas receive $T_1$ before $T_2$, then $T_1$ will be ordered no later than $T_2$.
\end{theorem}

\begin{proof}
    We denote by $G_{r_1}$ and $G_{r_2}$ in which $T_1$ and $T_2$ are ordered, respectively. If $r_1 < r_2$, then obviously $T_1$ is ordered before $T_2$ in the final transaction ordering. If $r_1 = r_2$, then from Lemma~\ref{lemma:earliergraph} we know that $node(d_1)$ is in the same or an earlier \emph{SCC} than $node(d_2)$. Thus, $T_1$ is ordered no later than $T_2$ in the final transaction ordering.

    If $r_1 > r_2$, then $node(d_2)$ is \emph{shaded}. Otherwise, assuming $node(d_2)$ is \emph{solid}, then there are at least $\n-\f$ \emph{committed local ordering} containing $d_2$ when $node(d_2)$ is added into $G_{r_2}$. From Lemma~\ref{lemma:edgedirection} we know that $G_{r_2}.weights[(d_2,d_1)] < \frac{\n-\f}{2}$, then $G_{r_2}.weights[(d_1,d_2)] > \frac{\n-\f}{2}$, and edge $e(d_1,d_2)$ exists. Therefore, $node(d_1)$ should be ordered in $G_{r_2}$ as it has a path to a \emph{solid} node, contradicting the fact that $r_1 > r_2$. Thus, $node(d_2)$ must be \emph{shaded} in $G_{r_2}$. Hence, there is some \emph{solid} node $node(d_s)$ such that there is a path $p_{2}$ from $node(d_2)$ to $node(d_s)$. We denote by $p_{1}$ the path in $G_{r_1}$ from the first ordered node to $node(d_1)$. Combining the following information:
    
    \begin{itemize}
        \item from Lemma~\ref{lemma:earliergraph} we know that all transactions on $p_1$, including $T-1$, are dependent on transaction $T_s$ of $node(d_s)$;
        \item from Lemma~\ref{lemma:dependent} we know that $T_2$ is dependent on $T_1$. 
        \item along path $p_2$, each transaction is dependent on the previous one, from $T_s$ to $T_2$.
    \end{itemize}

    Thus, the transactions on $p_1$ and $p_2$ form a \emph{cyclic dependent batch} $b$. For the transactions that are ordered in $G_{r_2}$ later than $T_s$, they are in the same \emph{SCC} as $T_s$ and then can be added into $b$. Therefore, even if $r_1 > r_2$, $T_1$ and $T_2$ are in the same \emph{cyclic dependent batch} in the final ordering, i.e., $T_1$ is ordered \emph{no later} than $T_2$.
\end{proof}
}

\section{Evaluation}\label{sec:eval}

We evaluate \FairDAG{} and \FairDAGRL{} by comparing their performance with other baseline protocols. 
We implemented the protocols~\cite{codebase} in Apache
ResilientDB (Incubating)~\cite{apache-resdb, geobft}. 
Apache ResilientDB is an open-source incubating blockchain project that supports various consensus protocols. It provides a fair comparison of each protocol by offering a high-performance framework. Researchers can focus
solely on their protocols without considering system structures
such as the network and thread models. We set up our experiments on CloudLAB m510 machines with 64 vCPUs and 64GB of DDR3 memory. Each replica and client run on a separate machine.

We compared \FairDAG{} and \FairDAGRL{} with the following baseline protocols:

\begin{itemize}
    \item \PBFT{}~\cite{pbftj}: A single-leader consensus protocol without fairness guarantees, $\n> 3\f$.
    \item \Pompe{}~\cite{pompe}: an absolute fairness protocol running on top of \PBFT{}, $\n> 3\f$.
    \item \Themis{}~\cite{themis}: a relative fairness protocol running on top of \PBFT{}, $\n > \frac{\f(2\gamma+2)}{2\gamma-1}$.
    \item \RCC{}~\cite{rcc}: a multi-proposer protocol that runs concurrent \PBFT{} instances without fairness guarantees, $\n> 3\f$.
    \item \Tusk{}~\cite{narwhal}, a \Changed{multi-proposer} DAG-based consensus protocol without fairness guarantees, $\n> 3\f$.
\end{itemize}

For \Themis{} and \FairDAGRL{}, we set $\gamma = 1$ in the experiments by default. And we implement the DAG layer of \FairDAG{} and \FairDAGRL{} on top of a variant of \Tusk{} with weak edges.

 \begin{figure}[t]
    \centering
    \setlength{\tabcolsep}{3pt}
    \scalebox{0.75}{\ref{tlbdlegend}}\\[5pt]
    \begin{tabular}{cc}
       \FairPerfTL &
    \end{tabular}
    \caption{Throughput vs latency with $\f=8$.}
    \label{fig:tps-vs-lat}
\end{figure}


\begin{figure*}[t]
     \centering
     \setlength{\tabcolsep}{3pt}
     \scalebox{0.75}{\ref{ftpslegend}}\\[5pt]
     \begin{tabular}{c c c c}
        \FairPerfFTPS\hspace{2mm}
        \FairPerfFLAT\hspace{2mm}
        \FairRegionTPS\hspace{2mm}
        \FairRegionLAT
     \end{tabular}
     \caption{Performance of \FDAG{} and baseline protocols, with varying $\f$ (a,b), and varying number of regions (c,d).}
     \label{fig:scalability}
 \end{figure*}

 \begin{figure}[t]
    \centering
    \setlength{\tabcolsep}{3pt}
    \scalebox{0.5}{\ref{rashnulegend}}\\[5pt]
    \begin{tabular}{cc@{\quad}cc}

    \rashnufigure{\datarashnutput}{(a) Throughput}{\Changed{skewness}}{Throughput (txn/s)}{0.01, 0.5, 0.99}{0.5}
    \rashnufigure{\datarashnulat}{(b) Latency}{\Changed{skewness}}{Latency (ms)}{0.01, 0.5, 0.99}{0.5}
    
    \end{tabular}
    
    \caption{\Changed{Performance of Rashnu-enhanced variants vs. relative fairness protocols.}}
    \label{fig:rashnu}
\end{figure}

\subsection{Scalability}~\label{sec:scalability}
In the scalability experiments, we measure two metrics:

\begin{itemize}
    \item {\em Throughput} -- the maximum number of transactions per second that the system reaches consensus.
    \item {\em Client Latency} -- the average duration between the time a client sends a transaction and the time the client receives $\f{+}1$ matching responses.
\end{itemize}

We compare the performance of the protocols with varying $\f$, the maximum number of faulty replicas allowed, from $5$ to $8$. With the same $\f$, different protocols have different minimum replica number requirements. For example, when $\f=5$, \Themis{} requires $\n=21$ replicas, while other protocols require $\n=16$ replicas.

Besides design, the performance of the protocols is highly related to the workloads. As shown in Figure~\ref{fig:tps-vs-lat}, where we set $\f=8$, as the workload increases, the throughput increases until the pipeline is fulfilled by the transactions. Then, after reaching the throughput limit, the latency increases as the workload increases. We define by \emph{optimal point} the point with the lowest latency while maintaining the highest throughput. And we evaluate their scalability at the \emph{optimal points} of the protocols with varying $\f$.

\begin{flushleft}
\textbf{Throughput.} Figure~\ref{fig:scalability} shows that
\Tusk{} and \RCC{} achieve higher throughput than other protocols because they have multiple proposers and \Changed{no overhead for fairness guarantees.} Due to the fairness overhead, when $\f=5$ and $\f=8$, \FairDAG{} reaches $83.5\%$ and $84.9\%$ throughput of \Tusk{}, while   \FairDAGRL{} reaches $11.9\%$ and $12.6\%$ throughput of \Tusk{}.
\end{flushleft}

However, compared to \Pompe{} and \Themis{}, the multi-proposer design of the DAG layer brings \FairDAG{} and \FairDAGRL{} advantages in throughput. When $\f = 5$ and $\f=8$, \FairDAG{} obtains $30.2\%$ and \Changed{$52.6\%$} higher throughput than \Pompe{}, respectively. Similarly, \FairDAGRL{} reaches \Changed{$7.5\%$ and $5.1\%$} higher throughput than \Themis{}.

\begin{flushleft}
\textbf{Latency.} 
Without the fairness overhead, \Tusk{} and \PBFT{}, as the underlying consensus protocols, have lower latency than the fairness protocols running on top of them. 
\end{flushleft}

With $\f = 5$ and $\f = 8$, \FairDAG{} latency is $7.1\%$ and $8.3\%$ higher than \Pompe{}, because \Tusk{}, the underlying DAG consensus protocol of \FairDAG{}, has a higher commit latency than \PBFT{}, the underlying consensus protocol of \Pompe{}{}. 

\FairDAGRL{} has a latency close to \Themis{} when $\f = 5$. As $\f$ grows, \FairDAGRL{} has a lower latency than \Themis{}, which is $20.9\%$ lower when $\f=8$. \FairDAGRL{} achieves a lower latency because \Themis{} needs $\f$ more correct replicas to guarantee fairness, which causes higher overhead for both consensus and ordering. By comparing the latency of \Themis{} with $\f=6$ and \FairDAGRL{} with $\f=8$, we can verify this claim: with the same replica number $\n=25$, \FairDAGRL{} achieves a $4.6\%$ higher latency than \Themis{}.

\textbf{Geo-distributed performance.} 
We conducted experiments under geo-distributed settings by deploying the systems across multiple AWS regions. Specifically, we varied the number of regions from 1 to 4. The regions include North Virginia, Oregon, London, and Zurich. We fixed $\f=8$ and deploys $\frac{\n}{k}$ replicas in each region, where $k$ is the number of regions.
Figure~\ref{fig:scalability} (c, d) show that in the geo-distributed setting, the latencies of all the protocols are high and increase with the number of regions, caused by the high inter-regional message delays. Furthermore, \FairDAG{} has higher throughput than the other fairness protocols. 

We found that for all protocols except the relative fairness protocols, increasing the batch size allowed us to achieve throughput values comparable to those in the single-region setting.  However, in \FairDAGRL{} and \Themis{}, a larger batch size leads to a higher overhead of the fairness layer, which increases quadratically with batch size. Moreover, while \FairDAGRL{} achieves only a $5.1\%$ throughput improvement over \Themis{} in a single-region setting, we observe that in the geo-distributed setting, \FairDAGRL{} outperforms \Themis{} by at least $42.1\%$. This significant gain is attributed to the robust performance of the underlying multi-proposer DAG-based consensus protocol in geo-scale settings with limited bandwidth and higher message delays.

\textbf{Data-dependent fairness.} \Rashnu{}~\cite{rashnu} proposes a technique to reduce the overhead of the fairness layer in relative fairness protocols by computing edge directions between only data-dependent transactions. This method is orthogonal to both \FairDAGRL{} and \Themis{}. We implementeded two \Rashnu{}-enhanced variants, called \textsc{Themis-Rashnu} and \textsc{FairDAG-RL-Rashnu}, and compared them to \Themis{} and \FairDAGRL{}. 

In this experiment, we implemented a transaction workload with keys following a Zipfian distribution. We fixed $\f=8$ and varied the skewness parameter $s$ from 0.01 to 0.99, where a higher skewness indicates greater data dependency between transactions. The results presented in Figure~\ref{fig:rashnu} show that the \Rashnu{} variants outperform the non-\Rashnu{} variants and perform better as the skewness decreases. This improvement stems from the reduced overhead in calculating edge directions when transactions are less interdependent.

\subsection{Tolerance to Byzantine Behavior}

Next, we will discuss the impact of Byzantine behaviors on the performance and transaction ordering fairness of the protocols.

\begin{figure}[t]
     \centering
     \setlength{\tabcolsep}{3pt}
    \scalebox{0.5}{\ref{tllegend}}\\[5pt]
     \begin{tabular}{cc}
        \AllTimeline 
     \end{tabular}
     \caption{Real-Time throughput with a faulty leader.}
     \label{fig:leaderfaulty}
\end{figure} 

\begin{flushleft}
    \textbf{Faulty leader.} In this experiment, we make the consensus leader replica in \Pompe{} and \Themis{} faulty, which would trigger a view-change for leader replacement. While for \FairDAG{} and \FairDAGRL{}, we make a replica faulty due to the multi-proposer design. Figure~\ref{fig:leaderfaulty} shows how the faulty leader affects the performance of \Themis{} and \Pompe{}. At time $7$, the faulty leader stops sending any messages. After a period without progress, a view-change is triggered to replace the faulty leader. At time $15$, the view-change is complete, and the throughput of \Pompe{} and \Themis{} recovers to the original level. In contrast, \FairDAGRL{} and \FairDAG{} are not affected because of the resilience provided by the multi-proposer design.
\end{flushleft}

\begin{flushleft}
    \textbf{Adversarial Manipulation.} We conduct two experiments in which Byzantine replicas attempt to manipulate transaction ordering. We evaluated the fairness quality of the fairness protocols under two Byzantine behaviors: (1) \emph{Reversing order}: in each round, the Byzantine replicas reverse the local orderings of the transactions it has received. (2) \emph{Targeted delay}: the Byzantine replicas intentionally delay targeted transactions by giving them higher local ordering indicators. In both cases, the Byzantine leader in \Pompe{} and \Themis{} excludes local orderings from $\f$ correct replicas.
\end{flushleft}

    For two transactions $T_1$ and $T_2$, we say that they are correctly ordered if $T_1$ is ordered before $T_2$ and: 
    \begin{enumerate}
        \item in relative fairness protocols, $weight[(d_1, d_2)] > \frac{n}{2}$;
        \item in absolute fairness protocols, the $\f{+}1$-th smallest local ordering indicator of $d_1$ is not larger than that of $d_2$.
    \end{enumerate} 
    For the \emph{reversing order} attack, we consider all transaction pairs; for the \emph{targeted delay} attack, we consider only the transaction pairs that involve the targeted transactions. 
    
    For \FairDAGRL{} and \Themis{}, we measure the ratio of correctly ordered pairs with different $Dist$ values, where $Dist(d_1, d_2) = \left| weight[(d_1, d_2)] - weight[(d_2, d_1)] \right|$. 
    
    For \FairDAG{} and \Pompe{}{}, we measure the ratio of correctly ordered pairs with different $Diff$ values, where $Diff(d_1, d_2)$ is the minimal number of Byzantine local ordering indicators needed to order $d_1$ before $d_2$. Formally, we denote by $asc\_ois^C_1[i]$ the $i$-th \emph{smallest} value in the correct local ordering indicators of $d_1$, and by $desc\_ois^C_2[j]$ the $j$-th \emph{largest} value in the correct local ordering indicators of $d_2$. To order $d_1$ before $d_2$, we need to find (1) a pair of $i$ and $j$ such that $asc\_ois^C_1[i] < desc\_ois^C_2[j]$, (2) $\f{+}1{-}i$ Byzantine local ordering indicators smaller than $asc\_ois^C_1[i]$, and (3) $\f{+}1{-}j$ Byzantine local ordering indicators larger than $desc\_ois^C_2[j]$. Thus, we have $Diff(d_1,d_2) = \min\{2(\f{+}1){-}i{-}j \mid asc\_ois^C_1[i] < desc\_ois^C_2[j]\}$, which is the minimal number of Byzantine local ordering indicators needed. We only consider transaction pairs with $Diff(d_1, d_2) > 0$.

    We set $\f=10$ and vary $\f_a$, the actual number of Byzantine replicas, from $0$ to $10$. For example, \Themis{}-7 denotes \Themis{} with $\f_a = 7$. As shown in Figure~\ref{fig:fairnessquality}, \FairDAGRL{} and \FairDAG{} consistently demonstrate better resilience against adversarial ordering manipulation in all experimental settings, compared to \Themis{} and \Pompe{}, respectively. The results substantiate our claim that \FDAG{} effectively mitigates adversarial ordering manipulation through the properties inherent in the DAG-based consensus layer.

\section{Related Work}\label{sec:related}

\fullversion{
\begin{flushleft}
\textbf{Fairness in BFT}    
\end{flushleft}
}

In traditional Byzantine Fault Tolerance (BFT) research, protocols~\cite{pbftj,bedrock,disth,poe} are designed to ensure both safety and liveness in the presence of malicious replicas. Although these protocols do not explicitly guarantee fair transaction ordering, they mitigate unfair ordering to some extent. Protocols such as HotStuff~\cite{hotstuff}, which employ leader rotation in a round-robin manner~\cite{hs2,beegees,jolteon,hs1,tendermint,spotless,gupta2025brief}, provide each participant with the opportunity to propose a block. Multi-proposer approaches, including concurrent consensus protocols~\cite{mirbft,spotless,rcc,giridharan2024autobahn} and DAG-based protocols~\cite{dagrider,narwhal,bullshark,shoal,shoalpp,mysticeti,fides}, enable multiple participants to propose blocks concurrently, ordering them globally through predetermined or randomized mechanisms. Although these protocols reduce the reliance on a single leader and distribute transaction ordering authority, a Byzantine participant can still manipulate the ordering of transactions within the blocks it proposes.

\begin{figure}[t]
    \centering
    \setlength{\tabcolsep}{3pt}
    \scalebox{0.5}{\ref{qualitylegend}}\\[5pt]
    \begin{tabular}{cc@{\quad}cc}

    \fairnessqualityfigure{\dataReverseOriginal}{(a) Reversing Order}{$Dist(d_1, d_2) = |weight[(d_1,d_2)]-weight[(d_2,d_1)]|$}{Ratio of Correctly Ordered Pairs}{\axisticksweight}{0.7}

    \fairnessqualityfigure{\dataReverse}{(b) Reversing Order}{$Dist(d_1, d_2)/\n$ $(\%)$}{Ratio of Correctly Ordered Pairs}{\axisticksquality}{0.7}\\

    \fairnessqualityfigure{\datafairnessqualityoriginal}{(c) \Changed{Targeted Delay}}{\Changed{$Dist(d_1, d_2)$}}{Ratio of Correctly Ordered Pairs}{\axisticksweight}{0}

    \fairnessqualityfigure{\datafairnessquality}{(d) \Changed{Targeted Delay}}{\Changed{$Dist(d_1, d_2)/\n$ $(\%)$}}{Ratio of Correctly Ordered Pairs}{\axisticksquality}{0}
    \end{tabular}
    \\[5pt]

    \scalebox{0.5}{\ref{qualitylegendab}}\\[5pt]
    \begin{tabular}{cc@{\quad}cc}

    \fairnessqualityfigureabsolute{\datareverseabsolute}{\Changed{(e) Reversing Order}}{\Changed{$Diff(d_1, d_2)$}}{Ratio of Correctly Ordered Pairs}{\axistickdiff}{0.5}
    \fairnessqualityfigureabsolute{\datafairnessqualityabsolute}{\Changed{(f) Targeted Delay}}{\Changed{$Diff(d_1, d_2)$}}{Ratio of Correctly Ordered Pairs}{\axistickdiff}{0.5}
    
    \end{tabular}
    
    \caption{Fairness quality of the \Changed{fairness protocols under adversarial ordering manipulation attacks.}}
    \label{fig:fairnessquality}
\end{figure}

Some protocols seek to eliminate the block proposers’ oligarchy over the ordering of transactions within blocks via \emph{censorship resistance}~\cite{encryptfair,cachin2001secure,honeybadger,fides,garimidi2025multiple,agarwal2025time, li2025transaction, malkhi2022maximal, zarbafian2023lyra,stathakopoulou2021adding}. In these protocols, a transaction is encrypted until the ordering of the transaction is determined. However, block proposers can still engage in censorship based on metadata, such as IP addresses, or prioritize their own transactions, knowing the content of encrypted transactions.

There are several prior fairness protocols that generate final ordering with collected local orderings. Wendy~\cite{wendyfairness} guarantees \emph{Timed-Relative-Fairness} similar to \emph{Ordering Linearizability}, but it relies on synchronized local clocks, which are impractical in asynchronous networks. DCN~\cite{constantinescu2023fair} reaches $\delta$-Median Fairness such that $T_1$ can be ordered before $T_2$ if $T_1$ is sent long enough earlier than $T_2$. Aequitas~\cite{aequitas} guarantees batch-order-fairness but suffers from liveness issues due to the existence of infinite \emph{Condorcet Cycles}, which \Themis{} solves via a \emph{batch unspooling} mechanism. Quick-Order-Fairness~\cite{quickorder} reaches batch-order-fairness with $\n>3\f$ replicas but incurs $O(\n^3)$ communication complexity for the consensus leader. Rashnu~\cite{rashnu} improves \Themis{} performance by guaranteeing \emph{$\gamma$-Batch-Order-Fairness} between only data-dependent transactions, but suffers from the same problems as \Themis{}. SpeedyFair ~\cite{mu2024separation} pipelines the consensus layer and fairness layer, but still relies on a single leader to collect local orderings.
\emph{Ambush} attacks are identified in~\cite{ambush}, which \FairDAGRL{} inherently mitigates with the underlying DAG-based consensus.

\fullversion{
\begin{flushleft}
\textbf{Multi-Proposer protocols}    
\end{flushleft}

A significant amount of research~\cite{shoal,mysticeti,shoalpp,sailfish} has been dedicated to reducing the latency in DAG-based protocols. These works employ various techniques, including \emph{pipelining DAG waves}, \emph{fast commit rules}, \emph{multi-anchor}, and \emph{uncertified DAG}, to enhance the efficiency and speed of these protocols.

Concurrent consensus protocols~\cite{rcc,mirbft,spotless} represent a distinct category of multi-proposer consensus protocols. These protocols run multiple concurrent consensus instances independently, generating a final ordering by globally ordering the committed blocks in each instance, round by round. However, these protocols face several challenges that make them unsuitable as the underlying consensus mechanisms for fairness protocols. First, the presence of a straggler instance, which lags behind the others, can substantially decrease throughput and increase system latency. Second, since the progress of different instances is not synchronized, a malicious leader of an instance could manipulate transaction ordering by delaying the block proposal until other replicas have proposed their blocks. DAG-based protocols address these issues effectively, as a replica can proceed to the next round once there are $\n-\f$ committed DAG vertices in the current round, ensuring more efficient and reliable progression.

}

\section{Conclusion}\label{sec:conclusion}

In this paper, we introduced \FDAG{}, a fair-ordering framework designed to run fairness protocols atop DAG-based consensus protocols. Through theoretical demonstration and experimental evaluation, we show that unlike previous fairness protocols, \FairDAG{} and \FairDAGRL{}, the two variants of \FDAG{}, not only uphold fairness guarantees, but also achieve better performance under normal and adversarial conditions, effectively constraining adversarial manipulation of transaction ordering.

\begin{acks}
This work is partially funded by NSF Award Number 2245373.
\end{acks}


\bibliographystyle{ACM-Reference-Format}
\balance
\bibliography{sources}

\end{document}